\documentclass{llncs}
\usepackage{amsmath}
\usepackage{amsfonts}
\usepackage[inline]{enumitem}
\usepackage{graphicx}
\usepackage{latexsym}
\usepackage{amssymb}
\usepackage{todonotes}
\usepackage{comment}
\usepackage[inline]{enumitem}
\usepackage{algorithm}
\usepackage{algpseudocode}
\usepackage{mathtools}
\usepackage{framed}
\usepackage{caption}
\usepackage{cite}
\usepackage{thmtools}
\usepackage{thm-restate}
\usepackage{hyperref}
\usepackage{boxedminipage}

\declaretheorem[name=Corollary]{cor}
\declaretheorem[name=Instantiation]{instant}

\DeclarePairedDelimiter{\floor}{\lfloor}{\rfloor}

\newcommand{\build}{\textup{\textsf{Build}}}
\newcommand{\lookup}{\textup{\textsf{Lookup}}}
\newcommand{\gen}{\textup{\textsf{Gen}}}
\newcommand{\found}{\textup{\textsf{found}}}
\newcommand{\calS}{\mathcal{S}}
\newcommand{\simulator}{\textup{\textsf{Sim}}}
\newcommand{\view}{\textup{\textsf{View}}}

\makeatletter
\newcommand{\algmargin}{\the\ALG@thistlm}
\makeatother
\newlength{\whilewidth}
\algdef{SE}[parWHILE]{parWhile}{EndparWhile}[1]
{\parbox[t]{\dimexpr\linewidth-\algmargin}{%
		\hangindent\whilewidth\strut\algorithmicwhile\ #1\ \algorithmicdo\strut}}{\algorithmicend\ \algorithmicwhile}%
\algnewcommand{\parState}[1]{\State%
	\parbox[t]{\dimexpr\linewidth-\algmargin}{\strut #1\strut}}

\spnewtheorem{claim}{Claim}{\itshape}{\rmfamily}

\begin{document}
\setlength{\textfloatsep}{7pt}
\setlength{\intextsep}{7pt}
\setlength{\abovecaptionskip}{7pt}
\title{Sub-logarithmic Distributed Oblivious RAM with Small Block Size\thanks{A first technical report was published on arXiv.org e-Print archive as  	arXiv:1802.05145 [cs.CR].}}

%\author{Anonymous}
%\institute{}

\author{Eyal Kushilevitz \and Tamer Mour}
\institute{Computer Science Department, Technion, Haifa 32000, Israel\\
\email{\{eyalk,tamer.mour\}@cs.technion.ac.il}}

\maketitle
\begin{abstract}
	\emph{Oblivious RAM} (ORAM) is a cryptographic primitive that allows a client to securely execute RAM programs over data that is stored in an untrusted server. \emph{Distributed Oblivious RAM} is a variant of ORAM, where the data is stored in $m>1$ servers. Extensive research over the last few decades have succeeded to reduce the bandwidth overhead of ORAM schemes, both in the single-server and the multi-server setting, from $O(\sqrt{N})$ to $O(1)$. However, all known protocols that achieve a sub-logarithmic overhead either require heavy server-side computation (e.g. homomorphic encryption), or a large block size of at least $\Omega(\log^3 N)$.
	
	In this paper, we present a family of distributed ORAM constructions that follow the hierarchical approach of Goldreich and Ostrovsky\cite{GO}. We enhance known techniques, and develop new ones, to take better advantage of the existence of multiple servers. By plugging efficient known hashing schemes in our constructions, we get the following results:
	\begin{enumerate}
		\item For any number $m\geq 2$ of servers, we show an $m$-server ORAM scheme with $O(\log N/\log\log N)$ overhead, and block size $\Omega(\log^2 N)$. This scheme is private even against an $(m-1)$-server collusion.
		\item A three-server ORAM construction with $O(\omega(1)\cdot\log N/\log\log N)$ overhead and a block size almost logarithmic, i.e. $\Omega(\log^{1+\epsilon}N)$.
	\end{enumerate}
	
	We also investigate a model where the servers are allowed to perform a linear amount of light local computations, and show that constant overhead is achievable in this model, through a simple four-server ORAM protocol. Through the theoretical lens, this is the first ORAM scheme with asymptotic constant overhead, and polylogarithmic block size, that does not use homomorphic encryption. Practically speaking, although we do not provide an implementation of the suggested construction, evidence from related work (e.g.\cite{DS}) makes us believe that despite the linear computational overhead, the construction can be potentially very efficient practically, in particular when applied to secure computation.
\end{abstract}
\section{Introduction}
Since it was first introduced by Goldreich and Ostrovsky\cite{GO}, the \emph{Oblivious RAM} problem has attracted a lot of attention (see, e.g. \cite{TREE,KLO,CIRCUIT}). Throughout the past three decades, efficient ORAM protocols were constructed (e.g. \cite{GM,PORAM}), their various applications, such as secure storage\cite{OS,AKST}, secure processors\cite{PROCESSORS}, and secure multi-party computation\cite{MPCORAM,LO},  were studied, and their limits were considered\cite{GO,PIRORAM,LN}.

\paragraph{Standard Model.} The standard ORAM model considers a setting where a client outsources his data to an untrusted server that supports read and write operations only. The goal of an \emph{ORAM simulation} is to simulate any RAM program that the client executes over the remote data, so that the same computation is performed, but the view of the server during the interaction would provide no information about the client's private input and the program executed, except their length. Clearly, encryption can be employed to hide the \emph{content} of the data, but the sequence of reads and write locations might leak information as well. Thus, the focus of ORAM is to hide the \emph{access pattern} made to the server.  The main metric considered in ORAM research is the \emph{bandwidth overhead} of an ORAM scheme (shortly referred to as ``overhead"), which is the multiplicative increase in the amount of communication incurred by an oblivious simulation relative to a regular run of the simulated program. In this standard model, researchers have been able to improve the overhead from $O(\log^3 N)$\cite{GO} to $O(\log N)$\cite{PORAM,CIRCUIT,OPTORAMA}, where $N$ is the number of data blocks in storage, and thus reaching the optimal overhead in that model due to the matching impossibility results of Goldreich and Ostrovsky\cite{GO} and Larsen and Nielsen\cite{LN}.

In an attempt to achieve sub-logarithmic overhead, research has deviated from the standard model (e.g. \cite{LO,AKST,WGK}). For instance, by allowing the server to perform some local computation, multiple works\cite{AKST,BUCKET,ONION} could achieve a constant overhead. However, this improvement comes at a cost: the server performs heavy homomorphic encryption computation which practically becomes the actual bottleneck of such schemes.

\paragraph{Distributed Oblivious RAM.} Another interesting line of work, often referred to as \emph{Distributed Oblivious RAM}\cite[etc.]{ZMZQ,PIRORAM,WGK}, was initiated by Ostrovsky and Shoup\cite{OS} and later refined by Lu and Ostrovsky\cite{LO}, and considers the multi-server setting. We denote by \emph{$(m,t)$-ORAM} an ORAM scheme that involves $m>1$ servers, out of which $t<m$ servers might collude. In the two-server setting, Zhang et al.\cite{ZMZQ} and Abraham et al.\cite{PIRORAM} construct $(2,1)$-ORAMs with sub-logarithmic overhead. In order to achieve $O(\log_d N)$ overhead (for any $d\in \mathbb{N}$) using their construction, Abraham et al. require that the size of a memory block, i.e. the data unit retrieved in a single query to the RAM, is $\Omega(d\log^2 N)$ (with larger blocks the asymptotic overhead increases). For example, for an overhead of $O(\log N/\log\log N)$, one has to work with blocks of relatively large size of $\Omega(\log^3 N)$, which may be undesired in many applications. In the work of Zhang et al., a polynomial block size of $\Omega(N^\epsilon)$ is required for a constant bandwidth blowup. Other attempts to achieve low overhead in the multi-server setting\cite{CHF} were shown to be vulnerable to concrete attacks\cite{PIRORAM}. These recent developments in distributed ORAM raise the following question, which we address in this paper:

\begin{center}
	\emph{Can we construct a sub-logarithmic distributed ORAM with a small block size?}
\end{center}

Known sub-logarithmic ORAMs\cite{PIRORAM,ZMZQ} belong to the family of \emph{tree-based ORAMs}\cite{TREE}. One of the key components in tree-based ORAMs is a \emph{position map} that is maintained through a recursive ORAM. Such a recursion imposes the requirement for a large block size (usually polylogarithmic)\footnote{To the best of our knowledge, the only tree-based ORAM that bypasses recursion belongs to Wang et al.~\cite{WGK}, which works in a different model where linear sever work is allowed (see preceding discussion).}. Thus, it seems that a positive answer to the question above will come, if at all, from constructions belonging to the other well-studied type of ORAMs, those based on the hierarchical solution of\cite{GO}. By applying the hierarchical approach to the distributed setting, Lu and Ostrovsky\cite{LO} succeeded to construct the first logarithmic \emph{hierarchical} ORAM scheme. In this paper, we show how to take a further advantage of the multiple servers in order to beat the logarithmic barrier, and still use a relatively small block size, with constructions in both the two-server and three-server settings. In addition, we consider the case where $t>1$, and show how to generalize our two-server solution to an $(m,m-1)$-ORAM, with the same asymptotic complexity, for any $m>2$.

\begin{table}[t]
	\small
	\begin{minipage}{\textwidth}
		\centering
		\begin{tabular}{ | l || c | c | c | c | c | c |}
			\hline
			Scheme	& $m$	& $t$ & Overhead	& Block size $B$& C. Strg. & S. Work \\ \hline\hline
			\cite{O,GO}&1 &-&$O(\log^3{N})$ & $\Omega(\log N)$ & $O(1)$ & -	\\ \hline
			\cite{KLO}&1&-&  $O(\frac{\log^2{N}}{\log\log N})$ &	$\Omega(\log N)$ & $O(1)$ & -  \\ \hline
			\cite{CIRCUIT}&1&-& $O(\log N\cdot\omega(1))$ & $\Omega(\log^2 N)$	& $O(1)$ & - \\ \hline
			\cite{OPTORAMA}&1&-& 	$O(\log N)$ & $\Omega(\log N)$  & $O(1)$ & - \\ \hline
			\cite{LO}& 2 & 1 & $O(\log{N})$ & $\Omega(\log N)$	& $O(1)$ & polylog  \\ \hline
			\cite{PERFECT3ORAM}& 3 & 1 & $O(\log^2{N})$ & $\Omega(\log N)$	& $O(1)$ & -  \\ \hline
			\cite{ZMZQ}& 2 & 1 & $O(1)$	& $\Omega(N^\epsilon)$ & $O(1)$ & polylog  \\ \hline
			\cite{PIRORAM}& 2	& 1	& $O(\log_d{N})$ &  $\omega(d\log^2{N})$& $O(1)$ & polylog\\ \hline
			\cite{DS}& 2 & 1 & $O(\sqrt{N})$	& $\Omega(\log N)$ & $O(1)$ & linear  \\ \hline
			\cite{WGK}& 2 & 1 & $O(\log N)$	& $\Omega(\log N)$ & $O(1)$ & linear  \\ \hline\hline
			\multicolumn{7}{|c|}{our 4-server construction}  \\ \hline
			Instantiation~\ref{inst:fss}& 4	& 1	& $O(1)$  & $\Omega(\lambda\log N)$ &$O(1)$ & linear \\ \hline\hline
			\multicolumn{7}{|c|}{our 3-server construction}  \\ \hline
			Instantiation~\ref{inst:three-klo}&3 &	1 	&  $O(\log_d{N}\cdot\omega(1))$ & $\Omega(d\log N)$ & $O(1)$ & polylog  \\
			\multicolumn{1}{|l||}{$d=\log^\epsilon N$}& &	 	&  $O(\frac{\log{N}}{\log\log{N}}\cdot \omega(1))$ & $\Omega(\log^{1+\epsilon} N)$ & &   \\ \hline
			Instantiation~\ref{inst:three-cgls}&3 &	 1	&  $O(\log_d{N})$ & $\Omega(d\log^{1.5} N)$ & $O(1)$ & polylog  \\
			\multicolumn{1}{|l||}{$d=\log^\epsilon N$}& &	 	&  $O(\frac{\log{N}}{\log\log{N}})$ & $\Omega(\log^{1.5+\epsilon} N)$ & &   \\ \hline\hline
			\multicolumn{7}{|c|}{our $m$-server construction}  \\ \hline
			Instantiation~\ref{inst:m-cgls}&$m\geq 2$ &	$m-1$ 	&  $O(\frac{\log N}{\log\log N})$ & $\Omega(\log^{2} N)$ & $O(1)$ & polylog  \\ \hline
		\end{tabular}
		\caption{Comparison of ORAM schemes. Columns by order: number of servers $m$, collusion size $t$, bandwidth overhead, block size $B$, client storage size in blocks (additive factors of the security parameter are omitted), and amortized server-side work per RAM step (computational overhead). $\lambda$ is a security parameter.}
		\label{table:results}
	\end{minipage}
\end{table}

\paragraph{Allowing Linear Server Work.} Although distributed ORAM schemes in the literature assume the servers are able to perform some server-side computations, most works limit them to be polylogarithmic in $N$ per read/write. An exception are two recent works\cite{DS,WGK} that investigate oblivious RAM when the servers are allowed to perform computations that are linear in $N$. In this stronger variant, the relatively new cryptographic primitive of \emph{function secret sharing} (FSS), introduced by Boyle et al.\cite{FSS}, was shown to be useful for constructing schemes that are practically efficient\cite{DS}, or that are of low interaction\cite{WGK}. However, none of the mentioned schemes achieve sub-logarithmic overhead, thus leaving us with the following question:

\begin{center}
	\emph{How efficient can distributed ORAM be if the servers do linear work?}
\end{center}

We show that by allowing the servers to perform linear computations per RAM step, we can achieve a distributed ORAM scheme with a small constant overhead. We restrict ourselves to simple server-side computations which were shown to perform well in practical implementations\cite{DS}. Therefore, we believe that our construction has the potential to be very efficient in practical applications, such as those to secure computation, as discussed below.

\paragraph{Application to Secure Computation.} Besides potentially being more efficient, distributed oblivious RAM was shown to be a useful tool in constructing secure computation protocols\cite{OS,MPCORAM,LO}. By using our distributed ORAM constructions as the underlying ORAM schemes in the elegant protocol from\cite{LO}, we get better asymptotic parameters for secure RAM computation than any known single-server or multi-server ORAM solution. Furthermore, our constructions are the first to be applicable to multi-party computation protocols that are secure against collusions.

\subsection{Our Contribution and Technical Overview}

\subsubsection{Sub-logarithmic Distributed ORAM Constructions.} Our main contribution is a family of distributed hierarchical ORAM constructions with any number of servers. Our constructions make a black-box use of hashing schemes. Instantiating our constructions with known hashing schemes, that were previously used in ORAM constructions~\cite{GM,KLO,LO,CGLS}, yields state-of-the-art results (see Table~\ref{table:results}). We elaborate below.
\paragraph{A Three-server ORAM Protocol.} By using techniques from \cite{LO} over the balanced hierarchy from\cite{KLO}, and using two-server PIR\cite{PIR} as a black box, we are able to construct an efficient $(3,1)$-ORAM scheme. Instantiating the scheme with cuckoo hash tables (similarly to \cite{GM,KLO,LO}) achieves an overhead of $O(\omega(1)\cdot \log_d N)$ with a block size of $B=\Omega(d\log N)$. Thus, for any $\epsilon>0$, we achieve $O(\omega(1)\cdot \log N/\log\log N)$ overhead with $B=\Omega(\log^{1+\epsilon}N)$.

In the classic hierarchical solution from~\cite{GO}, the data is stored in $\log N$ levels, and the protocol consists of two components: \emph{queries}, in which target virtual blocks are retrieved, and \emph{reshuffles}, which are performed to properly maintain the data structure. Roughly speaking, in a query, a single block is downloaded from every level, resulting in $\log N$ overhead per query. The reshuffles cost $\log N$ overhead per level, and $\log^2 N$ overall. Kushilevitz et al.\cite{KLO} suggest to balance the hierarchy by reducing the number of levels to $\log N/\log\log N$. In the balanced hierarchy, however, one has to download $\log N$ blocks from a level in every query. Thus, balancing the hierarchy ''balances'', in some sense, the asymptotic costs of the queries and reshuffles, as they both become $\log^2 N/\log\log N$.

At a high level, we carefully apply two-server techniques to reduce the overhead, both of the queries and the reshuffles, in the single-server ORAM of\cite{KLO}. More specifically, to reduce the queries cost, we use two-server PIR to allow the client to efficiently read the target block from the $\log N$ positions, it had otherwise have to download, from every level. By requiring the right (relatively small) block size, the cost of PIRs can be made constant per level and, therefore, $\log N/\log\log N$ in total. To reduce the reshuffles cost, we replace the single-server reshuffles with the cheaper two-server reshuffles, that were first used by Lu and Ostrovsky\cite{LO}, and that incur only a constant overhead per level.

So far, it sounds like we are already able to achieve $\log N/\log\log N$ overhead using two servers only. However, combining two-server PIR and two-server reshuffles is tricky: each assumes a different distribution of the data. In standard two-server PIR, the data is assumed to be identically replicated among the two servers. On the other hand, it is essential for the security of the two-server reshuffles from~\cite{LO} that every level in the hierarchy is held only by one of the two servers, so that the other server, which is used to reshuffle the data, does not see the access pattern to the level. We solve this problem by combining the two settings using three servers: every level is held only by two of the three servers in a way that conserves the security of the two-server reshuffles and, at the same time, provides the required setting for two-server PIR.
	
\paragraph{An $(m,m-1)$-ORAM Protocol.} We take further advantage of the existence of multiple servers and construct, for any integer $m\geq 2$, an $m$-server ORAM scheme that is private against a collusion of up to $m-1$ servers. Using oblivious two-tier hashing\cite{CGLS}, our scheme achieves an overhead of $O(\log N/\log\log N)$, for which it requires $B=\Omega(\log^2 N)$ (see Theorem~\ref{thrm:multi} and Instantiation~\ref{inst:m-cgls}).

We begin by describing a $(2,1)$-ORAM scheme, then briefly explain how to extend it to any number of servers $m>2$. Let us take a look back at our three-server construction. We were able to use both two-server PIR and two-server reshuffles only using a three-server setting. Now that we restrict ourselves to using two servers, we opt for the setting where the two servers store identical replicates of the entire data structure. Performing PIR is clearly still possible, but now that the queries in all levels are made to the same two servers, we cannot perform Lu and Ostrovsky's\cite{LO} two-server reshuffles securely. Instead, we use \emph{oblivious sort} (or, more generally, oblivious hashing) to reshuffle the levels. Oblivious sort is a sorting protocol in the client-server setting, where the server involved learns nothing about the obtained order of blocks. Oblivious sort is used in many single-server hierarchical ORAMs (e.g.~\cite{GO,KLO}), where it incurs $\log N$ overhead per level. Since we aim for a sub-logarithmic overhead, we avoid this undesired blowup by performing oblivious sort over the tags of the blocks only (i.e. their identities) which are much shorter, rather than over the blocks themselves. We require a block size large enough such that the gap between the size of the tags and the size of the blocks cancels out the multiplicative overhead of performing oblivious sort. Once the tags are shuffled into a level, it remains to match them with the blocks with the data. That is where the second server is used. We apply a secure two-server ''matching procedure'' which, at a high level, lets the second server to randomly permute the data blocks and send them to the server holding the shuffled tags. The latter can then match the data to the tags in an oblivious manner. Of course, the data exchange during the matching has to involve a subtle cryptographic treatment to preserve security.

The above scheme can be generalized to an $(m,m-1)$-ORAM, for any $m>2$. The data is replicated in all servers involved, and $m$-server PIR is used. The matching procedure is extended to an $m$-server procedure, where all the servers participate in randomly permuting the data.

\subsubsection{Constant Overhead with Linear Server Work.} We also investigate the model where linear server work is allowed, and show that constant overhead is achievable in this model (see Table~\ref{table:results}). The proposed scheme, described below, applies function secret sharing over secret-shared data, thus avoiding the need for encrypting the data using symmetric encryption (unlike existing schemes e.g. \cite{DS,WGK}).

\paragraph{A Simple Four-server ORAM Protocol.} Inspired by an idea first suggested in~\cite{OS}, we combine private information retrieval (PIR)\cite{PIR}, and PIR-write\cite{OS}, to obtain a four-server ORAM. To implement the PIR and PIR-write protocols efficiently, we make a black-box use of \emph{distributed point functions} (DPFs)\cite{DPF,FSS}, i.e. function secret sharing schemes for the class of point functions. Efficient DPFs can be used to construct \begin{enumerate*}[label=(\roman*)]\item a (computational) two-server PIR protocol if the data is replicated among the two servers, or\item a two-server PIR-write protocol for when the data is additively secret-shared among the two servers \end{enumerate*}. These two applications of DPFs are combined as follows: we create two additive shares of the data, and replicate each share twice. We send each of the four shares (two pairs of identical shares) to one of the four servers. A read is simulated with two instances of PIR, each invoked with a different pair of servers holding the same share. A write is simulated with two instance of PIR-write, each invoked with a different pair of servers holding different shares. 

We stress that the client in all of our constructions can be described using a simple small circuit, and therefore, our schemes can be used to obtain efficient secure multi-party protocols, following \cite{LO}. We elaborate on this in Appendix~\ref{app:mpc}. 

\subsection{Related Work}\label{sec:relatedwork}
\paragraph{Classic Hierarchical Solution.} The first hierarchical ORAM scheme appeared in the work of Ostrovsky \cite{O} and later in\cite{GO}. In this solution, the server holds the data in a hierarchy of levels, growing geometrically in size, where the $i^{th}$ level is a standard hash table with $2^i$ buckets of logarithmic size, and a hash function $h_i(\cdot)$, which is used to determine the location of blocks in the hash table: block of address $v$ may be found in level $i$ (if at all) in bucket $h_i(v)$. The scheme is initiated when all blocks are in the lowest level. An access to a block with a virtual address $v$ is simulated by downloading bucket $h_i(v)$ from every level $i$. Once the block is found, it is written back to the appropriate bucket in the smallest level ($i=0$). As a level fills up, it is merged down with the subsequent (larger) level $i+1$, which is reshuffled with a new hash function $h_{i+1}$ using oblivious sorting. Thus, a block is never accessed twice in the same level with the same hash function, hence the obliviousness of the scheme. Using AKS sorting network\cite{AKS} for the oblivious sort achieves an $O(\log^3 N)$ overhead.

\paragraph{Balanced Hierarchy.} Up until recently, the best known single-server ORAM scheme for general block size, with constant client memory, were obtained by applying an elegant ''balancing technique'' to the hierarchy of\cite{GO}, that reduces the number of levels in the hierarchy, in exchange for larger levels. The technique was first suggested by Kushilevitz et al.\cite{KLO}. Their scheme achieves an overhead of $O(\log^2/\log\log N)$, using \emph{oblivious cuckoo hashing} (first applied to ORAM in \cite{ORAMCUCKOOPR,GM}). An alternative construction, recently proposed by Chan et al.\cite{CGLS}, follows the same idea, but replaces the relatively complex cuckoo hashing with a simpler oblivious hashing that is based on a variant of the two-tier hashing scheme from\cite{TWOTIER}.

\paragraph{Tree-based ORAM.} Another well-studied family of ORAM schemes is tree-based ORAMs (e.g.\cite{TREE,PORAM,CIRCUIT,ONION,PIRORAM}). Tree-based ORAMs, as the name suggests, are oblivious RAM schemes where the data is stored in a tree structure. The first ORAM constructions with a logarithmic overhead, in the single-server model, were achieved following the tree-based approach \cite{PORAM,CIRCUIT}. However, tree-based ORAMs require relatively large block size of at least $B=\Omega(\log^2 N)$ (see Table~\ref{table:results}).
\setlength\tabcolsep{1.5pt}

\paragraph{Optimal ORAM with General Block Size.} The recent work of Asharov et al.\cite{OPTORAMA}, which improves upon the work of Patel et al.\cite{PANORAMA}, succeeds to achieve optimal logarithmic overhead with general block size (due to known lower bounds\cite{GO,LN}). Both results are based on the solution from~\cite{GO} and use non-trivial properties of the data in the hierarchy to optimize the overhead. 

\paragraph{Distributed ORAM Constructions.} Ostrovsky and Shoup\cite{OS} were the first to construct a distributed private-access storage scheme (that is not read-only). Their solution is based on the hierarchical ORAM from\cite{GO}. However, their model is a bit different than ours: they were interested in the amount of communication required for a single query (rather than a sequence of queries), and they did not limit the work done by the servers. Lu and Ostrovsky\cite{LO} considered the more general ORAM model, defined in Section~\ref{sec:model}. They presented the first two-server oblivious RAM scheme, and achieved a logarithmic overhead with a logarithmic block size by bypassing oblivious sort, and replacing it with an efficient reshuffling procedure that uses the two servers.

The tree approach was also studied in the multi-server model. Contrary to the hierarchical schemes, known distributed tree-based ORAMs\cite{ZMZQ,PIRORAM} beat the logarithmic barrier. The improvement in overhead could be achieved by using $k$-ary tree data structures, for some parameter $k=\omega(1)$. However, these constructions suffer from a few drawbacks, most importantly, they require a large polylogarithmic (sometimes polynomial) block size.

\paragraph{Amortized vs. Worst-Case Overhead.} Although overhead is defined as the amortized blowup in bandwidth, several works also considered the worst-case overhead per query. While tree-based schemes typically achieve low worst-case overhead\cite{TREE,PORAM,CIRCUIT}, hierarchical ORAMs might need additional work to turn their amortized overhead into worst-case overhead. Ostrovsky and Shoup\cite{OS} were the first to propose a method to de-amortize hierarchical constructions, which was adapted in some consequent works (e.g. \cite{GM,KLO}).

\paragraph{ORAMs with Linear Server-side Computational Overhead.} The work of Ostrovsky and Shoup~\cite{OS}, as well as some recent works\cite{DS,WGK} have considered the model where the servers are allowed to perform a linear amount of light computations. Both the works of Doerner and Shelat\cite{DS} and Wang et al.\cite{WGK} elegantly implement techniques from the standard model (square-root construction, and tree structure, respectively), and use the efficient PIR protocol from\cite{FSS}, to construct practically efficient two-server ORAM schemes with linear server-side computational overhead and bandwidth overhead matching their analogues in the single-server setting (see Table~\ref{table:results}).

\subsection{Paper Organization}
Section~\ref{sec:pre} contains a formal definition of the model and problem, as well as a description of the cryptographic tools we use in our constructions. In Section~\ref{sec:4servers}, we present our simple four-server ORAM with linear server work. In Section~\ref{sec:klo}, we provide a high-level description of the hierarchical ORAM framework, on which our main distributed ORAM constructions are based. In Sections~\ref{sec:3servers} and~\ref{sec:2multiservers}, we present these constructions, beginning with a high-level overview of each of them, followed by the full details and analysis. We finish with a few open questions. Due to page limit, de-amortization of our constructions, and a discussion of their application to secure computation, are left to Appendix~\ref{app:deamortize} and~\ref{app:mpc} (resp.), together with some additional complementary material.
\section{Preliminaries}\label{sec:pre}
\subsection{Model and Problem Definition}\label{sec:model}
\paragraph{The RAM Model.} We work in the RAM model, where a RAM machine consists of a CPU that interacts with a (supposedly remote) RAM storage. The CPU has a small number of local registers, therefore it uses the RAM storage for computations over large data, by performing a sequence of reads and writes to memory locations in the RAM. A sequence of $\ell$ queries is a list of $\ell$ tuples $(op_1,v_1,x_1),\dots,(op_\ell,v_\ell,x_\ell)$, where $op_i$ is either \textsf{Read} or \textsf{Write}, $v_i$ is the location of the memory cell to be read or written to, and $x_i$ is the data to be written to $v_i$ in case of a \textsf{Write}. For simplicity of notation, we unify both types of operations into an operation known as an \emph{access}, namely ``\textsf{Read} then \textsf{Write}". Hence, the \emph{access pattern} of the RAM machine is the sequence of the memory locations and the data $(v_1,x_1),\dots,(v_\ell,x_\ell)$.

\paragraph{Oblivious RAM Simulation.} A \emph{(single-server) oblivious RAM simulation}, shortly \emph{ORAM simulation}, is a simulation of a RAM machine, held by a client as a CPU, and a server as RAM storage. The client communicates with the server, and thus can query its memory. The server is untrusted but is assumed to be \emph{semi-honest}, i.e. it follows the protocol but attempts to learn as much information as possible from its view about the client's input and program. We also assume that the server is not just a memory machine with I/O functionality, but that it can perform basic local computations over its storage (e.g. shuffle arrays, compute simple hash functions, etc.). We refer to the access pattern of the RAM machine that is simulated as the \emph{virtual} access pattern. The access pattern that is produced by the oblivious simulation is called the \emph{actual} access pattern. The goal of ORAM is to simulate the RAM machine correctly, in a way that the distribution of the view of the server, i.e. the actual access pattern, would look independent of the virtual access pattern.

\begin{definition}[ORAM simulation, informal]
	Let \textsf{RAM} be a RAM machine. We say that a (probabilistic) RAM machine \textsf{ORAM} is an \emph{oblivious RAM simulation} of \textsf{RAM}, if \begin{enumerate*}[label=(\roman*)]
		\item (correctness) for any virtual access pattern $\vec{y}:=((v_1,x_1),\dots,(v_\ell,x_\ell))$, the output of \textsf{RAM} and \textsf{ORAM} at the end of the client-server interaction is equal with probability $\geq 1-negl(\ell)$, and
		\item (security) for any two virtual access patterns, $\vec{y},\vec{z}$, of length $\ell$, the corresponding distribution of the actual access patterns produced by \textsf{ORAM}, denoted $\vec{\tilde{y}}$ and $\vec{\tilde{z}}$, are computationally indistinguishable.
	\end{enumerate*}
\end{definition}

An alternative interpretation of the security requirement is as follows: the view of the server, during an ORAM simulation, can be simulated in a way that is indistinguishable from the actual view of the server, given only $\ell$.

\paragraph{Distributed Oblivious RAM.} A \emph{distributed oblivious RAM simulation} is the analogue of ORAM simulation in the multi-server setting. To simulate a RAM machine, the client now communicates with $m$ semi-honest servers, rather than with a single server only.
\begin{definition}[Distributed ORAM Simulation, informal]\label{def:m-oram}
	An \emph{$(m,t)$-ORAM simulation} ($0<t<m$) is an oblivious RAM simulation of a RAM machine, that is invoked by a CPU client and $m$ remote storage servers, and that is private against a collusion of $t$ corrupt servers. Namely, for any two actual access patterns $\vec{y},\vec{z}$ of length $\ell$, the corresponding combined view of any $t$ servers during the ORAM simulation (that consists of the actual access queries made to the $t$ servers) are computationally indistinguishable.
\end{definition}
With the involvement of more servers, we can hope to achieve schemes that are more efficient as well as schemes that protect against collusions of servers.

\paragraph{Parameters and Complexity Measures.} The main complexity measure in which ORAM schemes compete is the \emph{bandwidth overhead} (or, shortly, \emph{overhead}). When the ORAM protocol operates in the ``balls and bins" manner\cite{GO}, where the only type of data exchanged between the client and servers is actual memory blocks, it is convenient to define the overhead as the amount of actual memory blocks that are queried in the ORAM simulation to simulate a virtual query to a single block. However, in general, overhead is defined as the blowup in the number of information bits exchanged between the parties, relative to a non-oblivious execution of the program. Following the more general definition, the overhead is sometimes a function of the block size $B$. Clearly, we aim to achieve a small asymptotic overhead with block size as small as possible.

Other metrics include the size of the server storage and the client's local memory (in blocks), and the amount and type of the computations performed by the servers (e.g. simple arithmetics vs. heavy cryptography). We note that all of these notions are best defined in terms of overhead, compared to a non-oblivious execution of the program, e.g. storage overhead, computational overhead, etc..

\subsection{Private Information Retrieval}
\emph{Private information retrieval (PIR)}\cite{PIR} is a cryptographic primitive that allows a client to query a database stored in a remote server, without revealing the identity of the queried data block. Specifically, an array of $n$ blocks $X=(x_1,\dots,x_n)$ is stored in a server. The client, with input $i\in[n]$, wishes to retrieve $x_i$, while keeping $i$ private. PIR protocols allow the client to do that while minimizing the number of bits exchanged between the client and server. PIR is studied in two main settings: single-server PIR, where the database is stored in a single server, and the multi-server setting, where the database is replicated and stored in all servers, with which the client communicates simultaneously. More specifically, an $(m,t)$-PIR is a PIR protocol that involves $m>1$ servers and that is secure against any collusion of $t<m$ servers. It was shown in\cite{PIR} that non-trivial single-server PIRs cannot achieve information-theoretic security. This is made possible with two servers already in multi-server PIR. Moreover, many known two-server PIRs (both information theoretic and computational, e.g.\cite{PIR,PIRBIW,PIRDG,FSS}) do not involve heavy server-side computation, like homomorphic encryption or number theoretic computations, as opposed to known single-server protocols (e.g.\cite{PIRKO,PIRCMS,PIRGR}).

\section{A Simple Four-Server ORAM with Constant Overhead}\label{sec:4servers}
We present our four-server ORAM protocol with constant bandwidth overhead and linear server-side computational overhead. The protocol bypasses the need for symmetric encryption as it secret-shares the data among the servers. We use distributed point functions~\cite{DPF} (see Section~\ref{sec:dpf} below) as a building block.

\begin{theorem}[Four-server ORAM]\label{thrm:four}
	Assume the existence of a two-party DPF scheme for point functions $\{0,1\}^n\to\{0,1\}^m$ with share length $\Lambda(n,m)$ bits. Then, there exists a $(4,1)$-ORAM scheme with linear\footnote{Up to polylogarithmic factors.} server-side computation overhead and bandwidth overhead of $O(\Lambda(\log N,B)/B)$ for a block size of $B=\Omega(\Lambda(\log N,1))$.
\end{theorem}

Instantiating our scheme with the DPF from~\cite{FSS} obtains the following.

\begin{instant}\label{inst:fss}
	Assume the existence of one-way functions. Then, there exists a $(4,1)$-ORAM scheme with linear server-side computation overhead and constant bandwidth overhead for a block size of $B=\Omega(\lambda \log N)$, where $\lambda$ is a security parameter.
\end{instant}

\subsection{Building Block: Distributed Point Functions}
\label{sec:dpf}
\emph{Distributed Point Functions} (DPF), introduced by Gilboa and Ishai\cite{DPF}, are a special case of the broader cryptographic primitive called \emph{Function Secret Sharing} (FSS)\cite{FSS}. Analogous to standard secret sharing, an FSS allows a dealer to secret-share a function $f$ among two (or more) participants. Each participant is given a share that does not reveal any information about $f$. Using his share, each participant $p_i$, for $i\in\{0,1\}$, can compute a value $f_i(x)$ on any input $x$ in $f$'s domain. The value $f(x)$ can be computed by combining $f_0(x)$ and $f_1(x)$. In fact, $f(x)=f_0(x)+f_1(x)$. Distributed point function is an FSS for the class of point functions, i.e., all functions $P_{a,b}:\{0,1\}^n\to \{0,1\}^m$ that are defined by $P_{a,b}(a)=b$ and $P_{a,b}(a')=0^m$ for all $a'\neq a$. Boyle et al.\cite{FSS} construct a DPF scheme where the shares given to the parties are of size $O(\lambda n + m)$, where $\lambda$ is a security parameter, that is the length of a PRG seed. We are mainly interested in the application of DPFs to PIR and PIR-write\cite{DPF,FSS}.

\subsection{Overview}
Similarly to the schemes of\cite{DS,WGK}, we apply DPF-based PIR\cite{FSS} to allow the client to efficiently read records from a replicated data. If we allow linear server-side computational overhead, the task of oblivious reads becomes trivial by using DPFs. The remaining challenge is how to efficiently perform oblivious writes to the data.

The core idea behind the scheme is to apply DPFs not only for PIR, but also for a variant of \emph{PIR-write}. PIR-write (a variant of which was first investigated in\cite{OS}) is the write-only analog of PIR. We use DPFs to construct a simple two-server PIR-write where every server holds an additive share of the data. Our PIR-write protocol is limited in the sense that the client can only modify an existing record by some difference of his specification (rather than specifying the new value to be written). If the client has the ability to read the record in a private manner, then this limitation becomes irrelevant.

We combine the read-only PIR and the write-only PIR-write primitives to obtain a four-server ORAM scheme that enables both private reads and writes. In the setup, the client generates two additive shares of the initial data, $X^0,X^1$ s.t. $X=X^0\oplus X^1$, and replicates each of the shares. Each of the four shares obtained is given to one of the servers. For a private read, the client retrieves each of the shares $X^0,X^1$, using the DPF-based PIR protocol, with the two servers that hold the share. For a private write, the proposed PIR-write protocol is invoked with pairs of servers holding different shares of the data (see illustration in Figure~\ref{fig:4server}).

We remark that our method to combine PIR and PIR-write for ORAM is inspired by the 8-server ORAM scheme presented in\cite{OS}, in which an elementary 4-server PIR-write protocol was integrated with the PIR from\cite{PIR}.

\begin{figure}
	\centering
	\setlength{\belowcaptionskip}{-15pt}
	\includegraphics[scale=0.5]{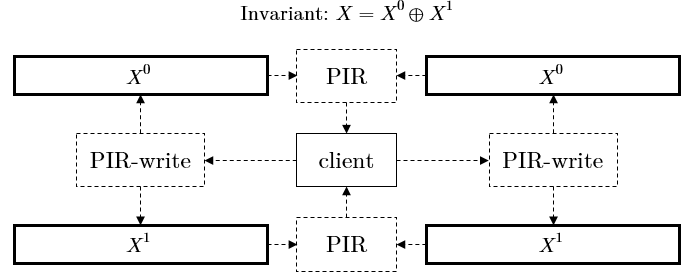}
	\caption{an illustration of the four-server construction.}
	\label{fig:4server}
\end{figure}

\subsection{Oblivious Read-only and Write-only Schemes with Two Servers}
\paragraph{Basic PIR and PIR-write.} Recall the classic two-server PIR protocol, proposed in\cite{PIR}. To securely retrieve a data block $x_i$ from an array $X=(x_1,\dots,x_N)$ that is stored in two non-colluding servers $\calS_0$ and $\calS_1$, the client generates two random $N$-bit vectors, $e^0_i$ and $e^1_i$ such that $e^0_i\oplus e^1_i=e_i$, where $e_i$ is the $i^{th}$ unit vector, and sends $e^b_i$ to $\calS_b$. In other words, the client secret-shares the vector $e_i$ among the two servers. Then, each server, computes the inner product $x^b_i:=X\cdot e^b_i$ and sends it to the client. It is easy to see that $x_i=x^0_i\oplus x^1_i$.

The same approach can be used for two-server PIR-write. However, now we require that the data is shared, rather than replicated, among the two servers. Namely, server $\calS_b$ holds a share of the data $X^b$, such that $X^0\oplus X^1=X$. In order to write a new value $\hat{x}_i$ to the $i^{th}$ block in the array, the client secret-shares the vector $(\hat{x}_i\oplus x_i)e_i$ to the two servers. Each of the servers adds his share to $X^b$, and obtains a new array $\hat{X}^b$. After this update, the servers have additive shares of $X$ with the updated value of $x_i$. Notice that we assume that the client already read and knows $x_i$; this is not standard in the PIR-write model.
\paragraph{Efficient PIR and PIR-write via DPFs.} In the heart of the PIR and PIR-write protocols described above is the secret sharing of vectors of size $N$. Applying standard additive secret sharing yields protocols with linear communication cost. Since we share a very specific type of vectors, specifically, unit vectors and their multiples, standard secret sharing is an overkill. Instead, we use DPFs. The values of a point function $P_{i,x}:[N]\to\{0,1\}^m$ (that evaluates $x$ at $i$, and zero elsewhere) can be represented by a multiple of a unit vector $v_{i,x}:=x e_i$. Hence, one can view distributed point functions as a means to ''compress'' shares of unit vectors and their multiples. We can use DPFs to share such a vector among two participants $p_0$ and $p_1$, as follows. We secret-share the function $P_{i,x}$ using a DPF scheme, and generate two shares $P^0_{i,x}$ and $P^1_{i,x}$. For $b\in\{0,1\}$, share $P^b_{i,x}$ is sent to participant $p_b$. The participants can compute their shares of the vector $v_{i,x}$ by evaluating their DPF share on every input in $[N]$. Namely, $p_b$ computes his share $v_{i,x}^b:=(P^b_{i,x}(1),\dots,P^b_{i,x}(n))$. From the correctness of the underlying DPF scheme, it holds that $v^0_{i,x}\oplus v^1_{i,x}=v_{i,x}$. Further, from the security of the DPF, the participants do not learn anything about the vector $v_{i,x}$ except the fact that it is a multiple of a unit vector. Using the DPF construction from\cite{FSS}, we have a secret sharing scheme for unit vectors and their multiples, with communication complexity $O(\lambda\log N + m)$, assuming the existence of a PRG $G:\{0,1\}^\lambda \to \{0,1\}^m$.

\subsection{Construction of Four-Server ORAM}
\paragraph{Initial Server Storage.} Let $\calS_0^0,\calS_1^0,\calS_0^1$ and $\calS_1^1$ be the four servers involved in the protocol. Let $X=(x_1,\dots,x_N)$ be the data consisting of $N$ blocks, each of size $B=\Omega(\Lambda(\log N,1))$ bits. In initialization, the client generates two additive shares of the data, $X^0=(x^0_1,\dots,x^0_N)$ and $X^1=(x^1_1,\dots,x^1_N)$. That is, $X^0$ and $X^1$ are two random vectors of $N$ blocks, satisfying $X^0\oplus X^1=X$. For $b\in\{0,1\}$, the client sends $X^b$ to both $\calS_0^b$ and $\calS_1^b$. Throughout the ORAM simulation, we maintain the following invariant: for $b\in\{0,1\}$, $\calS_0^b$ and $\calS_1^b$ have an identical array $X^b$, such that $X^0$ and $X^1$ are random additive shares of $X$.

\paragraph{Query Protocol.} To obliviously simulate a read/write query to the $i^{th}$ block in the data, the client first reads the value $x_i$ via two PIR queries: a two-server PIR with $\calS_0^0$ and $\calS_1^0$ to retrieve $x_i^0$, and a two-server PIR with $\calS_0^1$ and $\calS_1^1$ to retrieve $x_i^1$. The client then computes $x_i$ using the two shares. Second, to write a new value $\hat{x}_i$ to the data (which can possibly be equal to $x_i$), the client performs two \emph{identical} invocations of two-server PIR-write, each with servers $\calS_b^0$ and $\calS_b^1$ for $b\in\{0,1\}$. It is important that $\calS_0^b,\calS_1^b$ (for $b\in\{0,1\}$) receive an identical PIR-write query, since otherwise, they will no longer have two identical replicates.
\subsection{Analysis}
The security of the scheme follows directly from the security of the underlying DPF protocol from\cite{FSS}. It remains to analyze the bandwidth cost. To simulate a query, the client sends each of the servers two DPF shares: one for reading of length $\Lambda(\log N,1)$ bits, and another for writing of length $\Lambda(\log N, B)$. With a block size of $B=\Omega(\Lambda(\log N,1))$ this translates to $O(\Lambda(\log N, B)/B)$ bandwidth overhead. Each of the servers, in return, answers by sending two blocks.

\section{The Balanced Hierarchical ORAM Framework}\label{sec:klo}

In this section, we lay the groundwork for our constructions in the standard distributed ORAM model, that are presented later in Sections~\ref{sec:3servers} and~\ref{sec:2multiservers}.

\subsection{Main Building Block: Hashing}
Hashing, or more accurately, oblivious hashing, has been a main building block of hierarchical ORAM schemes since their first appearance in\cite{O}. Various types of hashing schemes, each with different parameters and properties, were plugged in ORAM constructions in an attempt to achieve efficient protocols (e.g.\cite{GO,GM,CGLS}). Hashing stands at the heart of our constructions as well. However, since we make a generic black-box use of hashing, we do not limit ourselves to a specific scheme, but rather take a modular approach.

We consider an \emph{$(n,m,s)$-hashing scheme}\footnote{Implicitly stated parameters may be omitted for brevity.}, $H$, to be defined by three procedures: $\gen$ for key generation, $\build$ for constructing a hash table $T$ of size $m$ that contains $n$ given data elements, using the generated key, and $\lookup$ for querying $T$ for a target value. The scheme may also use a stash to store at most $s$ elements that could not be inserted into $T$. In a context where a collection of hashing schemes operate simultaneously (e.g. ORAMs), a \emph{shared stash} may be used by all hash tables. We denote by $C_\build(H)$ and $C_\lookup(H)$, the build-up complexity and the query complexity of $H$ (resp.) in terms of communication (in the client-server setting).

An \emph{oblivious} hashing scheme is a scheme whose $\build$ and $\lookup$ procedures are oblivious of the stored data and the queried elements (respectively). In Appendix~\ref{app:hashing}, we provide formal definitions and notation for the above, and survey a few of the schemes that were used in prior ORAM works.

\subsection{Our Starting Point: The Single-server Scheme from\cite{KLO}}

\subsubsection{Overview.} The starting point of our distributed ORAM constructions in Sections~\ref{sec:3servers} and~\ref{sec:2multiservers} is the single-server scheme from\cite{KLO}. In standard hierarchical ORAMs, the server stores the data in $\log N$ levels, where every level is a hash table, larger by a factor of 2 than the preceding level. Kushilevitz et al. changed this by having $L=\log_d N$ levels, where the size of the $i^{th}$ level is proportional to $(d-1)\cdot d^{i-1}$. Having less levels eventually leads to the efficiency in overhead, however, since level $i+1$ is larger by a factor of $d$ (no longer constant) than level $i$, merging level $i$ with level $i+1$ becomes costly (shuffling an array of size $(d-1)\cdot d^{i}$ every $(d-1)\cdot d^{i-1}$ queries). To solve this problem, every level is stored in $d-1$ separate hash tables of equal size in a way that allows us to reshuffle every level into a single hash table in the subsequent level.

\begin{theorem}[\hspace{1sp}\cite{KLO,CGLS}]\label{thrm:klo}
	Let $d$ be a parameter, and define $L=\log_d N$. Assume the existence of one-way functions, and a collection $\{H_i\}^L_{i=1}$, where $H_i$ is an oblivious $(d^{i-1}k,\cdot,\cdot)$-hashing scheme, with a shared stash of size $s$. Then there exists a single-server ORAM scheme that achieves the following overhead for block size $B=\Omega(\log N)$.
	\[
	O\left(k+s+\sum_{i=1}^L d\cdot C_{\lookup}(H_i)+\sum_{i=1}^L \frac{C_\build(H_i)}{d^{i-1}k}\right)
	\]
\end{theorem}

A special variant of the theorem was proven by Kushilevitz et al.\cite{KLO}. In their work, they use a well-specified collection of hashing schemes (consisting of both standard and cuckoo hashing\cite{CUCKOO}), and obtain an overhead of $O(\log^2 N/\log\log N)$. The modular approach to hierarchical ORAM was taken by Chan et al.\cite{CGLS}, in light of their observations regarding the conceptual complexity of cuckoo hashing, and their construction of a simpler oblivious hashing scheme that achieves a similar result. Our results in the distributed setting fit perfectly in this generic framework, as they are independent of the underlying hashing schemes. Below, we elaborate the details of the construction from\cite{KLO}, as a preparation towards the following sections.

\paragraph{Data Structure.} The top level, indexed $i=0$, is stored as a plain array of size $k$. As for the rest of the hierarchy, the $i^{th}$ level ($i=1\dots L$) is stored in $d-1$ hash tables, generated by an oblivious $(d^{i-1}k,\cdot,\cdot)$-hashing scheme $H_i$. For every $i=1,\dots,L$ and $j=1,\dots,d-1$, let $T^{j}_i$ be the $j^{th}$ table in the $i^{th}$ level, and let $\kappa^j_i$ be its corresponding key. All hashing schemes in the hierarchy share a stash $S$\footnote{In the scheme of\cite{KLO}, the shared stash is 'virtualized', and is re-inserted into the hierarchy. We intentionally roll-back this optimization in preparation to our distributed constructions.}. The keys $\kappa^j_i$ can be encrypted and stored remotely in the server. Also, the client stores and maintains a counter $t$ that starts at zero, and increments by one after every virtual access is simulated. The ORAM simulation starts with the initial data stored entirely in the lowest level.

\paragraph{Blocks Positioning Invariant.} Throughout the ORAM simulation, every data block in the virtual memory resides either in the top level, or in one of the hash tables in the hierarchy, or in the shared stash. The blocks are hashed according to their virtual addresses. The data structure does not contain duplicated records.

\paragraph{Blocks Flow and Reshuffles.} Once a block is queried, it is inserted into the top level, therefore the level fills up after $k$ queries. Reshuffles are used to push blocks down the hierarchy and prevent overflows in the data structure. Basically, every time we try to insert blocks to a full level, we clear the level by reshuffling its blocks to a lower level. For instance, the top level is reshuffled every $k$ queries.

In every reshuffle, blocks are inserted into the first empty hash table in the highest level possible, using the corresponding $\build$ procedure, with a freshly generated key. Thus, the first time the top level is reshuffled (after round $k$), its blocks are inserted to the first table in the next level, i.e. $T^1_1$, which becomes full. The top level fills up again after $k$ queries. This time, the reshuffle is made to $T^2_1$, as $T^1_1$ is not empty anymore. After $d-1$ such reshuffles, the entire first level becomes full, therefore, after $d\cdot k$ queries, we need to reshuffle both the top level and the first level. This time, we insert all blocks in these levels into $T^1_2$.

Observe that this mechanism is analogous to the process of counting in base $d$: every level represents a digit, whose value is the number of full hash tables in the level. An increment of a digit with value $d-1$, equivalently - insertion to a full level, is done by resetting the digit to zero, and incrementing the next digit by 1, that is, reshuffling the level to a hash table in the next level. A demonstration of this analogy is given in Figure~\ref{fig:flow}. Using this observation, we formalize the process as follows: in every  round $t=t'\cdot k$, levels $0,\dots,i$ are reshuffled down to hash table $T^j_{i+1}$, where $i$ is the maximal integer for which $d^{i}\mid t'$, and $j=(t' \bmod d^{i+1}) / d^{i}$. Notice that level $i$ is reshuffled every $k\cdot d^{i}$ queries.

\begin{figure}
	\setlength{\belowcaptionskip}{-15pt}
	\includegraphics[width=\linewidth]{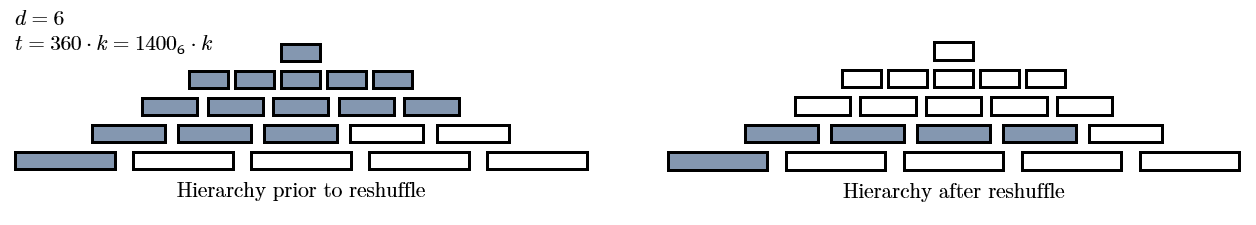}
	\caption{a demonstration of the flow of blocks during an ORAM simulation with $d=6$. A gray cell indicates a full hash table, a white one is an empty table.}
	\label{fig:flow}
\end{figure}

\paragraph{Query.} In order to retrieve a data block with virtual address $v$, the client searches for the block in the top level and the stash first. Then, for every level $i$, the client scans hash tables $T^j_i$ using $H_i.\lookup$ procedure, \emph{in reverse order}, starting with the table that was last reshuffled into. Once the target block was found, the scan continues with dummy queries. This is important for security (see Claim~\ref{claim:queriedonce}). Appendix~\ref{app:pseudocode} contains a detailed pseudo-code for the query algorithm.

\section{A Three-server ORAM Scheme}\label{sec:3servers}
Below, we formally state our first result in the standard distributed ORAM model: an efficient three-server ORAM scheme.
\begin{theorem}[Three-server ORAM using regular hashing]\label{thrm:three}
	Let $d$ be a parameter, and define $L=\log_d N$. Assume the existence of one-way functions, and a collection $\{H_i\}^L_{i=1}$, where $H_i$ is a $(d^{i-1}(k+s),m_i,s)$-hashing scheme. Then, there exists a $(3,1)$-ORAM scheme that achieves an overhead of
	\[
	O\left(k+L + \sum_{i=1}^L \frac{m_i}{d^{i-1}k}\right)
	\]	
	for block size $B=\Omega(\alpha d\log N+s\log d)$, where $\alpha:=\max_i C_\lookup(H_i)$.
\end{theorem}
We propose two different instantiations of our construction, each with a different collection of hashing schemes that was used in prior ORAM works\cite{GM,KLO,CGLS}. Both instantiations yield sub-logarithmic overhead, and their parameters are very close. However, Instantiation~\ref{inst:three-cgls} may be conceptually simpler (due to\cite{CGLS}). More details about the used hashing schemes can be found in Appendix~\ref{app:hashing}.

First, we plug in the collection of hashing schemes used by Goodrich and Mitzenmacher\cite{GM}, and later by Kushilevitz et al.\cite{KLO}. The collection mainly consists of cuckoo hashing schemes, however, since stashed cuckoo hashing was shown to have a negligible failure probability only when the size of the hash table is polylogarithmic in $N$ (specifically, $\Omega(\log^7 N)$)\cite{GM}, standard hashing with bucket size $\log N/\log\log N$ is used in the first $\Theta (\log_d \log N)$ levels. We point out that in both mentioned works\cite{GM,KLO}, the stash size for cuckoo hashing is logarithmic. In our instantiation, we use a stash of size $\Theta(\omega(1)\cdot\log N/\log\log N)$. Although \cite{GM} proved that failure probability is negligible in $N$ when the stash is of size $s=\Theta(\log N)$ and the size of the table is $m=\Omega(\log^7 N)$ (by extending the proof for constant stash size from\cite{CUCKOOSTASH}), their proof works whenever the value $m^{-\Theta(s)}$ is negligible in $N$, and in particular, when we choose $s=\Theta(\omega(1)\cdot\log N/\log\log N)$.

\begin{instant}[Three-server ORAM using cuckoo hashing]\label{inst:three-klo}
	Assume the existence of one-way functions. Let $d$ be a parameter at most polylogarithmic in $N$. Then, there exists a three-server ORAM scheme that achieves overhead of $O(\log_d N\cdot\omega(1))$ for $B=\Omega(d\log N)$.
\end{instant}

When $d=\log^\epsilon N$ for a constant $\epsilon\in(0,1)$, we achieve an overhead of $O(\omega(1)\cdot\log N/\log\log N)$ with $B=O(\log^{1+\epsilon} N)$.

Alternatively, we can use the simple two-tier hashing scheme from\cite{TWOTIER}, with buckets of size $\log^{0.5+\epsilon}N$, to achieve the following parameters.

\begin{instant}[Three-server ORAM using two-tier hashing]\label{inst:three-cgls}
	Assume the existence of one-way functions. Let $d$ be a parameter at most polylogarithmic in $N$. Then, there exists a three-server ORAM scheme that achieves overhead of $O(\log_d N)$ for block size $B=\Omega(d\log^{1.5+\epsilon} N)$.
\end{instant}
For $d=\log^\epsilon N$ , we obtain an overhead of $O(\log N/\log\log N)$ with $B=O(\log^{1.5+2\epsilon} N)$.

\subsection{Overview}
Our three-server scheme is based on the single-server balanced hierarchical structure of Kushilevitz et al.\cite{KLO} (described in Section~\ref{sec:klo}). We take advantage of the existence of multiple servers and reduce the overhead as follows.

\paragraph{Reduce query cost using PIR.} One of the consequences of balancing the hierarchy is having multiple hash tables in a level, in any of which a target block can reside. More specifically, if $T^1_i,\dots,T^{d-1}_i$ are the hash tables at level $i$, then a block with address $v$ can possibly reside in any of the positions in $T^j_i[H_i.\lookup(v,\kappa^j_i)]$ for $j=1,\cdots, d-1$. To retrieve such a block, we could basically download all blocks in these positions, i.e. $\sum_{i=1}^L(d-1)C_\lookup(H_i)$ blocks in total. This already exceeds the promised overhead. Instead, we use PIR to extract the block efficiently without compromising the security of the scheme. For every level $i$, starting from the top, we invoke a PIR protocol over the array that consists of the $(d-1)C_\lookup(H_i)$ possible positions for $v$ in the level.

Performing PIR queries requires that the client knows the exact position of the target block in the queried array, namely, in which bucket, out of the $d-1$ possibilities, block $v$ resides, if at all. Therefore, the client first downloads the addresses of all blocks in the array, and only then performs the PIR query. Although some PIR protocols in the literature (e.g.\cite{FSS}) do not impose this requirement, we still need to download the addresses since it is essential for the security of the protocol that the client re-writes the address of the queried block.

An address of a block can be represented using $\log N$ bits. Thus, downloading the addresses of all possible positions in all levels costs us $\sum_{i=1}^L(d-1)C_\lookup(H_i)\log N$ bits of communication. If we choose $B=\Omega(\alpha d\log N)$ for $\alpha=\max_i C_\lookup(H_i)$, this cost translates to the desirable $O(L)$ overhead. Two-server PIRs work in the model where the data is replicated and stored in two non-colluding servers. Thus, every level in the hierarchy, except the top level, will be stored, accessed, and modified simultaneously in two of the three servers.
\paragraph{Reduce reshuffles cost by bypassing oblivious hashing.} We use a variant of the reshuffle procedure suggested by Lu and Ostrovsky\cite{LO}. Their protocol  works in a model with two non-colluding servers, where one server stores the odd levels, and the other stores the even levels. Before reshuffling a level, the servers gather all blocks to be reshuffled, permute them randomly, and exchange them through the client, who re-encrypts them and tags them with pseudorandom tags. The level is then reshuffled by one server using some \emph{regular} hashing scheme (not necessarily oblivious), and is sent to the other server, record by record, through the client. The security of their scheme follows from the following observations: \begin{enumerate}[label=(\roman*),topsep=0pt]
	\item the blocks are re-encrypted and permuted randomly before the reshuffle, eliminating any dependency on prior events,
	\item the blocks are hashed according to pseudorandom tags, hence their order is (computationally) independent of their identities,
	\item\label{item:losec2} the server that holds a level cannot distinguish between dummy queries and real ones since he was not involved in the reshuffle, and
	\item \label{item:losec3} the server that reshuffles the level (and can tell a dummy query) does not see the accesses to the level at all.
\end{enumerate} Applying this method naively when each of the servers holds the entire hierarchy might reveal information about the access pattern since~\ref{item:losec2} and~\ref{item:losec3} no longer hold. Therefore, we should adapt their method wisely, while having two replicates of every level, to allow performing PIR queries. A straight forward implementation would require four servers: two holding replicates of the odd levels, and two holding replicates of the even levels. However, this can be done using three servers only by having every pair of servers (out of the three possible pairs) hold every third level.
\subsection{Full Construction}
\subsubsection{Data Structure.}
The data is virtually viewed as an array of $N$ blocks, each of size $\Omega(\alpha d\log N)$ bits. Every block therefore has a virtual address in $[N]$. 

\paragraph{Distributed Server Storage.} The data structure is identical to that from\cite{KLO}, however, our scheme uses three servers, $\calS_0,\calS_1$, and $\calS_2$, to store the data. The top level is stored in all servers. Every other level is held by two servers only: for $j=0,\dots,\floor{\frac{L}{3}}$, $\calS_0$ and $\calS_1$ share replicates of levels $i=3j$, $\calS_1$ and $\calS_2$ share replicates of levels $3j+1$, and $\calS_2$ and $\calS_0$ both hold all levels $i=3j+2$.

\paragraph{Dummy Blocks.} Dummy blocks are blocks that are not ''real'' (not part of the virtual memory), but are treated as such, and assigned dummy virtual addresses. From the point of view of the 'reshuffler' server, a dummy block, unlike an empty block, cannot be distinguished from a real block. We use two types of dummy blocks, both essential for the security of the scheme.
\begin{enumerate}[label=(\roman*),topsep=0pt]
	\item \emph{Dummy Hash Blocks}.  Dummy hash blocks replace real blocks once they are read and written to the top level. The security of our scheme relies on the fact that all blocks in the hierarchy are of distinct addresses, hence the importance of this replacement.
	\item \emph{Dummy Stash Blocks.} Dummy stash blocks are created by the client to fill in empty entries in the hierarchy. Since our scheme uses a stash to handle overflows, the number of blocks in the stash and in each of the hash tables is not deterministic and is dependent on the access pattern. To hide this information from the server that performs the reshuffling of a level, we fill all empty entries in the stash, and some of the empty entries in the hash tables, with dummy stash blocks.
\end{enumerate}

\paragraph{Block Headers.} To properly manage the data, the client needs to know the identity of every block it downloads (i.e. its virtual address). Therefore, every entry in the server storage contains, besides the data of the block, a header that consists of the virtual address of the block, which can be either an address in $[N]$, a numbered dummy address, such as '$dummyHash\circ t$' or '$dummyStash\circ r$', or just '$empty$'. The length of the header is $O(\log N)$ bits, thus does not affect the asymptotic block size. Unless explicitly stated otherwise, the headers are downloaded, uploaded and re-encrypted together with the data. An entry with a block of virtual address $v$ and data $x$ is denoted by the tuple $(v,x)$.

\paragraph{Tags.} Since we use the servers for reshuffling the levels, we wish to hide the virtual addresses of the blocks to be reshuffled. We use pseudorandom tags to replace these addresses, as first suggested in\cite{LO}. The tags are computed using a keyed PRF, $F_s$, that is known to the client only. When generating a new hash table, the server hashes the blocks according to their tags (rather than their virtual addresses). Furthermore, to eliminate any dependency between tags that are seen in different reshuffles, the client keeps an \emph{epoch} $e^j_i$ for every hash table $T^j_i$ in the hierarchy. The epoch of a table is updated prior to every reshuffle, and is used, together with $i$ and $j$, to compute fresh tags for blocks in the table. The epochs can be stored remotely in the servers to avoid large client storage.
\subsubsection{Protocol.} Again, we refer to the balanced hierarchy of\cite{KLO} (see Algorithm~\ref{algo:1query}) as our starting point.

\paragraph{Query.} We replace the reads performed by the client with PIR protocols that are executed over arrays in the data. Specifically, the first PIR is performed over the stash to retrieve the target block if it is found there. The top level can be downloaded entirely since it has to be re-written anyway. The search continues to the other levels in the hierarchy in the order specified in Section~\ref{sec:klo}. The target block can possibly reside in any of the $d-1$ hash tables in a level, therefore, the client invokes a PIR protocol to extract the target block out of the many possible positions. Every PIR in the procedure is preceded by downloading the headers in the queried array, using which the client knows the position of the target block. A technical detailed description is provided in Algorithm~\ref{algo:3query}.
\begin{algorithm}[t!b]
	\caption{Three-Server Construction: Query}
	\label{algo:3query}
	\begin{algorithmic}[1]
		\State Allocate a local register of the size of a single record.
		\State Initialize a flag $\found\gets 0$.
		\parState{Download the top level, one record at a time. If $v$ is found at some entry $(v,x)$ then store $x$ in the local register, and mark $found\gets 1$.}
		\parState{Download all headers from $S$. If $v$ was found among these headers, let $p$ be its position, and mark $\found\gets 1$. Otherwise, let $p$ be a position of a random entry in the stash. Invoke $\Call{PIR}{S,p}$ to fetch $(v,x)$ with any two of the three servers, and store $x$ in the register.}
		\For{every level $i=1\dots L$}
		\State $t' \gets \floor{t/k}$
		\State $r\gets \floor{(t' \bmod d^i) / d^{i-1}}$
		\State $headers\gets \emptyset$
		\For{every hash table $j=r\dots 1$}
		\parState{If $\found=false$, compute the corresponding tag of $v$, $\tau\gets F_s(i,j,e^j_i,v)$. Otherwise, assign $\tau\gets F_s(i,j,e^j_i,dummy \circ t)$.}\label{step:computetagsquery}
		\parState{$Q^j_i\gets H_i.\lookup(\tau,\kappa^j_i)$}
		\parState{Download all headers of entries in $T^{j}_i[Q^j_i]$, and append them to $headers$. If one of the headers says $v$, mark $\found\gets true$.}
		\EndFor
		\parState{Let $p$ be the position of $v$ in $headers$ if it was found there, or a random value in $\{1,\dots,|headers|\}$ otherwise.}
		\parState{Let $A$ be the array of entries corresponding to headers in $headers$.}
		\parState{Invoke \Call{PIR}{$A,p$} to fetch $(v,x)$ with the two servers holding level $i$, and store $x$ in the register (if $v$ was not found in $headers$ this would be a dummy PIR).}
		\parState{Re-encrypt $headers$, and upload it back to the two servers, while changing $v$ to $dummyHash\circ t$.}\label{step:reuploadheaders}
		\EndFor
		\State If the query is a write query, overwrite $x$ in the register.
		\State Read each entry of the entire top level from both servers one at a time, re-encrypt it, then write it back, with the following exception: if the entry $(v,x)$ was first found at the top level, then overwrite $x$ with the (possibly) new value from the register, otherwise, write $(v,x)$ in the first empty spot of the form $(empty,\cdot)$.
		\State Increment the counter $t$, and reshuffle the appropriate levels.
	\end{algorithmic}
\end{algorithm}
\paragraph{Reshuffles.} Let $\calS_a$ and $\calS_b$ be the two servers holding level $i+1$, and let $\calS_c$ be the other server. Reshuffling levels $0,\dots,i$ into hash table $T_{i+1}^j$ is performed as follows. As a first step, we send all non-empty blocks that should be reshuffled (including stash) to $\calS_c$, by having the servers exchange the blocks they hold in levels $0,\dots, i$ and the stash, through the client, one block at a time, in a random order. Besides forwarding the blocks to $\calS_c$, the client also re-encrypts every block and re-tags it with a fresh tag (using epochs, as already mentioned). Once $\calS_c$ has all tagged blocks, he can create a new hash table and stash using the appropriate $\build$ procedure. He then sends the hash table and stash, one record at a time, to the client. The client re-encrypts all records, and forwards them to the other two servers, who store the hash table in $T^j_{i+1}$, and the stash to its place. The client uses dummy stash blocks to replace as many empty blocks as needed to get a full hash table, and a full stash. This is important since we do not want to reveal the load of the stash to the server that does the next reshuffle. The reshuffle procedure is described in full details in Algorithm~\ref{algo:3reshuffle}.

\begin{algorithm}[t!b]
	\caption{Three-Server Construction: Reshuffle}
	\label{algo:3reshuffle}
	\begin{algorithmic}[1]
		\item[\textbf{Reshuffling into table $T^j_{i+1}$}]
		\item[Let $\calS_a$ and $\calS_b$ be the servers holding level $i+1$, and let $\calS_c$ be the other server.]
		\parState{Every server of the three allocates a temporary array. For every level $\ell$ between levels $1$ and $i$, let $\calS^\ell$ be the server with the smallest id that holds level $\ell$. For every such $\ell$, $\calS^\ell$ inserts all records in level $\ell$ to its temporary array. In addition, one of the servers, say $\calS_0$, inserts all stash records into its temporary array.}
		\parState{$\calS_c$ applies a random permutation on its temporary array, and sends the records one by one to the client. The client re-encrypts each record and sends it to $\calS_b$. $\calS_b$ inserts all records it receives to its array.  $\calS_b$ permutes its array randomly, and forwards it to $\calS_a$ through the client (who re-encrypts them). $\calS_a$, in his turn, also inserts all received records, applies a random permutation, and sends them one by one to the client.}		
		
		\parState{The client re-encrypts every non-empty record $(v,x)$ and sends it to $\calS_c$, together with a tag, which is the output of the PRF $F_s(i + 1,j, e^j_{i+1}, v)$, where $e_{i+1}$ is the new epoch of $T^j_{i+1}$. Note that $v$ may be a virtual memory address, or a dummy value. In this step, dummy records are treated as real records and only empty records are discarded.}\label{step:computetagsreshuffle}
		\parState{$\calS_c$ receives $d^{i}(k +s)$ tagged records, which are all records that should be reshuffled into $T^j_{i+1}$. It generates a new key $\kappa^j_i\gets H_i.\gen(N)$, and constructs a hash table and a stash $(T^j_i,S)\gets H_i.\build(\kappa^j_i,Y)$, where $Y$ is the set of tagged records received from the client. If the insertion fails, a new key is generated (this happens with a negligible probability). $\calS_c$ then informs the client about the number of elements inside the stash, $\sigma$, and the key $\kappa^j_i$, then sends both the hash table $T^j_i$ and the stash one record at a time to the client.}
		\parState{As the client receives entries from $\calS_c$ one at a time, it re-encrypts each record and sends it to both $\calS_a$ and $\calS_b$ without modifying the contents except: \begin{enumerate}[label=(\alph*)]
				\item The first $\sigma$ empty records in the table the client receives from $\calS_c$ are encrypted as ($dummyStash\circ r$, $\cdot$), incrementing $r$ each time.
				\item Subsequent empty records from the table are encrypted as $(empty, \cdot)$.
				\item Every empty record in the stash is re-encrypted as $(dummyStash\circ r, \cdot)$, incrementing $r$ each time.
			\end{enumerate}
		\label{step:fullstash}
		\parState{$\calS_a$ and $\calS_b$ store the table records in level $i + 1$ in the order in which they were received, and store the stash
		records at the top level.}}
	\end{algorithmic}
\end{algorithm}

\subsection{Analysis}\label{sec:3analysis}
\subsubsection{Complexity.} We begin with analyzing the complexity of the described scheme.
\paragraph{Server storage.} The combined server storage contains a stash of size $s$, a top level of size $k$, and two duplicates of every other level $i$, consisting of $d-1$ hash tables of size $m_i$ each. In total, we have $O\left(s+k+\sum_{i=1}^L dm_i\right)$.
\paragraph{Client storage.} The client uses constant working memory as it only receives and forwards records, one record at a time. Notice that the client does not need to keep all the headers he downloads prior to executing PIR queries, and it would be sufficient to keep only the position of the target block among them. 

\paragraph{Overhead.} We now analyze the cost of performing a single query. First, consider the communication cost of downloading the headers for the PIRs. The PIRs are performed over the stash and each of the levels $i=1,\dots,L$. The number of headers downloaded amounts to $s+\sum_{i=1}^{L}(d-1)C_\lookup(H_i)\leq s+\alpha L(d-1)$, which is equivalent to $O(L)$ blocks of the required minimum size. Overall, $L+1$ PIR queries are invoked. For levels $i=1,\dots,L$, the PIR queries are performed over arrays of size at most $(d-1)C_\lookup(H_i)$. By using even the classic two-server PIR from\cite{PIR}, this costs $(d-1)C_\lookup(H_i)<\alpha d$ bits and a single block per level. The stash adds $s$ bits and a block. All of this sums up to no more than $O(L)$ data blocks. As for the actual blocks the client downloads, these are the blocks of the top level, $O(k)$ in number.

It remains to add the overhead caused by the reshuffles made between queries. Blocks are reshuffled down to some hash table in the $i^{th}$ level if $i$ is the smallest integer for which $(t/k)\bmod d^i\neq 0$. This occurs whenever $t/k$ is a multiple of $d^{i-1}$, but not of $d^i$, i.e., at most once every $k\cdot d^{i-1}$ queries. One can clearly see that during the reshuffle of a hash table $T^j_i$, the number of blocks transmitted is asymptotically bounded by the size of $T^j_i$ and the size of the stash, which is $m_i+s=O(m_i)$. Hence, the amortized overhead of the reshuffles is $O(\sum_{i=1}^L \frac{m_i}{d^{i-1}k})$.

\subsubsection{Security.} Next, we present the security proof for our construction. We prove that the access pattern to any of the servers in the scheme is oblivious and independent on the input. We describe a simulator $\simulator_a$ (for $a\in\{0,1,2\}$), that produces an output that is computationally indistinguishable from the view of server $\calS_a$ during the execution of the protocol, upon any sequence of virtual queries $v_1,\dots,v_\ell$, given only its length $\ell$.

\begin{lemma}[Security of the three-server ORAM]\label{lem:3sec}
	Let $\view_a(\vec{y})$ be the view of server $\calS_a$ during the execution of the three-server ORAM protocol, described in Algorithms~\ref{algo:3query} and~\ref{algo:3reshuffle}, over a virtual access pattern $\vec{y}=((v_1,x_1),\dots,(v_\ell,x_\ell))$. There exist simulators $\simulator_0,\simulator_1,\simulator_2$, such that for every $\vec{y}$ of length $\ell$, and every $a\in\{0,1,2\}$, the distributions $\simulator_a(\ell)$ and $\view_a(\vec{y})$ are computationally indistinguishable.
\end{lemma}

\paragraph{Proof Sketch.} As in all previous works, we assume that the client uses one-way functions to encrypt and authenticate the data held in the servers, and therefore, encrypted data is indistinguishable by content (notice that the client re-encrypts every piece of data before sending it). We replace the keyed tagging functions, that are modeled as PRFs in the scheme, with random functions. These preliminary steps can be formalized using proper standard hybrid arguments, which we avoid for brevity.

We begin by inspecting the view of the servers during the reshuffle procedure. The procedure starts with the servers exchanging all blocks stored in levels $1,\dots,i$ and in the stash, and sending them to $\calS_c$. It is essential for security that the number of these blocks is independent of the input, as we show below.

\begin{claim}\label{claim:stash}
	Throughout the ORAM simulation, the stash is always full (contains $s$ records).
\end{claim}
\begin{proof}
	In the reshuffle procedure (Algorithm~\ref{algo:3reshuffle}, Step~\ref{step:fullstash}), every empty slot in the stash is filled with a dummy stash block, thus the claim holds.\qed
\end{proof}
\begin{restatable}{claim}{reshufflingclaim}
	\label{claim:reshuff}
	Let $t$ be a multiple of $k$, and denote $t'=t/k$. For every $1\leq i\leq L$, define $r^t_i:=\floor{(t'\bmod d^i)/d^{i-1}}$. Then,
	\begin{enumerate}[label=(\roman*)]
		\item\label{item:reshufftop} the top level is full prior to the reshuffle at round $t$, and is empty afterwards.
		\item for every other level $1\leq i\leq L$, once the reshuffle is completed, the first $r^t_i$ tables in level $i$ (i.e., $T^1_i,\dots,T^{r^t_i}_i$) are full (contain $d^i(k+s)$ records each), and all other tables in level $i$ are empty.
	\end{enumerate}
\end{restatable}

Recall the analogy of the reshuffles to counting in base $d$ (see Section~\ref{sec:klo}). Notice that $r^t_i$ can be also defined as the $i^{th}$ digit in the base $d$ representation of $t'$. Claim~\ref{claim:reshuff} follows from these two observations. Due to space limitation, a complete proof for this claim is provided in Appendix~\ref{app:claims}.

Having shown that the amount of data exchanged during the first steps of the reshuffling procedure depends only on $t$, we can simulate the view of any of the servers by a sequence of arbitrary encrypted data of the appropriate length. Next, $\calS_c$ receives $(k+s)\cdot d^i$ tagged encrypted records (Claim~\ref{claim:reshuff}). In Claim~\ref{claim:3uniquetags} below, we show that these records have unique addresses, therefore their tags will also be unique (with overwhelming probability). Furthermore, these tags are computed using a random function that has not been used before (fresh epoch). Hence, the view of $\calS_c$ can be simulated as a sequence of $(k+s)\cdot d^i$ arbitrary encrypted records with random distinct tags. Once $\calS_c$ successfully creates the hash table, it sends it to $\calS_a$ and $\calS_b$ via the client. The size of the hash table is fixed. Again, since the entries of the hash tables are encrypted, they are indistinguishable from any arbitrary sequence of encrypted records, and can be simulated as such.

\begin{claim} \label{claim:3uniquetags}
	At all times during the execution, any non-empty record of the form $(v,\cdot)$ will appear at most once in all hash tables in the hierarchy.
\end{claim}
\begin{proof}
	A label $v$ can be either a virtual address or a dummy label of some type. A real virtual address can be added to the top level of the hierarchy once it is queried, in such a case, it is removed from its prior location and replaced by a dummy hash block. Dummy blocks, of any type, are labeled with a counter that is incremented after the creation of each block, thus cannot reoccur.\qed
\end{proof}

\begin{cor} \label{cor:3uniquetagsreshuff} The tagging function $F_s(\cdot)$ will not be computed twice on the same input throughout the executions of Algorithm~\ref{algo:3reshuffle} during the ORAM simulation.
\end{cor}
\begin{proof}
	 Otherwise, the same $v$ would have appeared twice during the pass in Step~\ref{step:computetagsreshuffle} of Algorithm~\ref{algo:3reshuffle}, in contradiction to Claim~\ref{claim:3uniquetags}.\qed
\end{proof}

To summarize, to simulate the view of the servers during the reshuffling phase, $\simulator_a(\ell)$ and $\simulator_b(\ell)$ output a sequence of encrypted arbitrary records of the appropriate length (which is fixed due to Claims~\ref{claim:stash} and~\ref{claim:reshuff}), whereas $\simulator_c(\ell)$ outputs a sequence of encrypted arbitrary records that are tagged using distinct uniform values ($a,b,c$ alternate between 0,1,2 throughout the phases). From Corollary~\ref{cor:3uniquetagsreshuff} and the security of the underlying symmetric encryption and PRFs, these outputs are indistinguishable from the views of the servers at the reshuffles.

We proceed to simulating the access pattern during queries. A query for a block $v$ begins, independently of $v$, with downloading all blocks in the top level, and all headers in the stash. Next, a PIR is invoked over the stash. From the assumed security of the underlying PIR scheme, there exist two simulators $\simulator^{\text{PIR}}_0(m),\simulator^{\text{PIR}}_1(m)$, that simulate the individual views of the two servers (resp.) involved in the protocol, given only the size of the queried array, $m$. We use these simulators to simulate the view of the servers involved in the PIR over the stash, and in the other PIR invocations to follow.

It remains to show that the identity of the blocks over which the PIRs are called, i.e. the values $Q^j_i$ that a server $\calS_a$ sees during the execution of Algorithm~\ref{algo:3query}, can be simulated as well. Recall that, at every execution of the algorithm, $Q^j_i$ is computed, for every $i,j$, as $H_i.\lookup(\tau,\kappa^j_i)$, where $\tau$ is a tag computed using $F_s$, and $\kappa^j_i$ is the used hash key. We denote by $\langle Q^j_i\rangle_a$ the sequence of $Q^j_i$ values seen by $\calS_a$ at \emph{all} executions of Algorithm~\ref{algo:3query} during the ORAM simulation (these values correspond to levels $i$ that are stored in $\calS_a$). We also denote by $\langle\tau\rangle_a$ and $\langle\kappa^j_i\rangle_a$ the values used to compute $\langle Q^j_i\rangle_a$.

\begin{claim}\label{claim:queriedonce}
	The same $v$ will not be queried upon twice at the same hash table (in Algorithm~\ref{algo:3query}) between two reshuffles of the table during the ORAM execution.
\end{claim}
\begin{proof}
	For a dummy address $v=$'$dummy\circ t$', the lemma is trivial, since it can be queried only at access $t$. If a virtual address $v$ is queried at hash table $T_i^j$, it will be written to the top level. If $v$ reaches level $i$ before $T_i^j$ is shuffled down to level $i+1$, then it would be placed in a hash table $T_i^{j'}$ for $j'>j$. Therefore, subsequent queries to $v$ will find $v$ either at some level above level $i$, or at hash table $T_i^{j'}$. In both cases, the client will perform a dummy query at hash table $T_i^j$ (hence the importance of the reverse order scan over the hash tables).\qed
\end{proof}

\begin{cor}\label{cor:3uniquetagsqueries} The tagging function $F_s$ will not be computed twice on the same input throughout the executions of Algorithm~\ref{algo:3query} during the ORAM simulation.
\end{cor}
\begin{proof}
	Consider two different inputs $(i,j,e^j_i,v), (i',j',{e^j_i}',v')$. If $i=i',j=j'$ and $e^j_i={e^j_i}'$, then, due to Lemma~\ref{claim:queriedonce}, $v\neq v'$.\qed
\end{proof}

\begin{claim}\label{claim:3tagsareuniform}
	The sequence $\langle \tau \rangle_a$, defined above, is computationally indistinguishable from a uniform sequence of unique tags, given the view of $\calS_a$ during the reshuffles.
\end{claim}
\begin{proof}
	The claim follows from Corollary~\ref{cor:3uniquetagsqueries}, and from the fact that the tags $\langle \tau \rangle_a$ correspond to levels that are not hashed by $\calS_a$. The view of $\calS_a$ during the reshuffles of such levels consist of encrypted data only, and is computationally independent of $\langle \tau \rangle_a$ (from the security of the underlying encryption scheme). The view of $\calS_a$ during the reshuffles of any other level are also independent of $\langle \tau \rangle_a$ since they consist of encrypted data, and random tags that are computed for different values of $i$.\qed
\end{proof}

\begin{claim}\label{claim:3keysareuniform}
	The sequence $\langle \kappa^j_i \rangle_a$, defined above, is computationally indistinguishable from a uniform sequence of hash keys, given the view of $\calS_a$ during the reshuffles.
\end{claim}
\begin{proof}
	The values in $\langle \kappa^j_i \rangle_a$ are uniformly chosen by a server other than $\calS_a$ under the constraint that keys do not cause a failure in the build up of the corresponding hash tables. From the definition of a hashing scheme (Definition~\ref{def:hash}), a failure occurs with a negligible probability over the choice of the key. The distribution of $\langle \kappa^j_i \rangle_a$ is clearly computationally independent of the view of $\calS_a$ during the reshuffles, which consists of encrypted data, and random tags at levels in which the hash keys $\langle \kappa^j_i \rangle_a$ are not used.\qed
\end{proof}

To conclude, in Claims~\ref{claim:3tagsareuniform} and~\ref{claim:3keysareuniform}, we show that $\langle\tau\rangle_a$ and $\langle\kappa^j_i\rangle_a$ (resp.) are indistinguishable from sequences of uniformly chosen values, given the view obtained at the reshuffles. Therefore, to simulate the values $\langle Q^j_i\rangle_a$, the simulator $\simulator_a(\ell)$ computes the output of $H_i.\lookup$ for uniformly random sequence of tags and hash keys. This completes the proof of Lemma~\ref{lem:3sec}, and hence Theorem~\ref{thrm:three}.

\section{A Family of Multi-Server ORAM Schemes}\label{sec:2multiservers}

In this section, we show how to use oblivious hashing, rather than regular hashing, to obtain our sub-logarithmic $(m,m-1)$-ORAM, for any $m>2$.

\begin{theorem}[Multi-server ORAM using oblivious hashing]\label{thrm:multi}
	Let $d$ be a parameter, and define $L=\log_d N$. Assume the existence of one-way functions, and a collection $\{H_i\}^L_{i=1}$, where $H_i$ is an oblivious $(d^{i-1}(k+s),m_i,s)$-hashing scheme. Then, for any $m\geq 2$, there exists an ($m$,$m-1$)-ORAM scheme that achieves the following overhead for block size $B=\Omega(\beta\log N+\alpha d\log N)$ \[O\left(k+L+ \sum_{i=1}^L \frac{m_i}{d^{i-1}k}\right)\] where $\alpha:=\max_i C_\lookup(H_i)$ and $\beta:=\max_i \frac{C_\build(H_i)}{d^{i-1}k}$.
\end{theorem}

For instantiating the multi-server construction, we suggest using the oblivious variant of the hashing schemes from\cite{GM,KLO,CGLS} (see Appendix~\ref{app:hashing}). Here, in contrary to the construction in Section~\ref{sec:3servers}, two-tier hashing obtains (slightly) better results, both in overhead and in the minimal required $B$ (see Table~\ref{table:results}).
\begin{instant}[$m$-server ORAM using two-tier hashing]\label{inst:m-cgls}
	Assume the existence of one-way functions. Then, for any $m\geq 2$, there exists a ($m$,$m-1$)-ORAM scheme that achieves overhead of $O(\log N/\log\log N)$ for block size of $B=\Omega(\log^2 N)$.
\end{instant}

For clarity of presentation, we first present the special case of our construction in the two-server setting. We then show how to generalize the construction to the setting where more servers are involved. In Appendix~\ref{app:deamortize}, we show how to de-amortize the construction to obtain a worst-case ORAM scheme.

\subsection{Two-Server ORAM: Overview}\label{sec:2servers}
Our two-server ORAM solution is based on the three-server scheme from Section~\ref{sec:3servers}. We make the following modifications to reduce the number of servers.

\paragraph{Back to oblivious hashing.} Now that we limit ourselves to using two servers only, each of which has to hold a replicate of the data for the PIR queries, we lose the ability to perform the reshuffles through a ''third-party''. Hence, we require now that the underlying hashing schemes are oblivious, and the build-up of the hash table is done using the oblivious $\build$ procedures, where the client is the CPU, and one of the servers takes the role of the RAM.

Recall that the tags were essential for the security of the three-server scheme since the reshuffles were made by one of the servers, to which we did not want to reveal the identity of the blocks being reshuffled. Now that the reshuffling is done using oblivious hashing that hides any information about the records that are being hashed, or the hash keys used to hash them, using tags is not necessary anymore. Instead, the blocks are hashed, and accessed, by their headers.
\paragraph{Optimizing the reshuffles.} Creating a hash table at level $i$ using $H_i.\build$, incurs an overhead of $C_\build(H_i)$ when naively applied. We observe that in any hashing scheme (by Definition~\ref{def:hash}), the only input relevant for the build-up of a hash table is the tags or, in our case, the headers of the blocks being reshuffled. Based on this observation, we suggest the following solution. The reshuffles are modified so that the build-up of the hash tables is given, as input, the set of headers, rather than the blocks themselves. Since the headers are smaller than the blocks by a factor of at least $\beta:=\max_i \frac{C_\build(H_i)}{d^{i-1}k}$ (see Theorem~\ref{thrm:multi}), the overhead incurred by the build-ups is cut by $\beta$, making it linear in $d^{i-1}k$.
\paragraph{Finalizing with a matching procedure.} As the headers are hashed, we still have to move the data to the new hash table. To securely match the data to the hashed headers, we suggest a method that involves tagging the data elements, and permuting them randomly by the servers. Thus, tags are still used, however, in a totally different context.
\subsection{Two-Server ORAM: Full Construction}
\paragraph{Data Structure.}We start with the scheme from Section~\ref{sec:3servers}. The server storage remains as is, except the data is now distributed among two servers, rather than three. All levels in the hierarchy, as well as the stash, are duplicated and stored in $\calS_0$ and $\calS_1$. The protocol guarantees that, at the end of every round, the data in the two servers will be identical. 

\paragraph{Query.} Every virtual access is simulated as described in Algorithm~\ref{algo:3query}, with the exception that the target block is queried upon in the hash tables by its virtual address, rather than its tag: $H_i.\lookup(v,\kappa^j_i)$ rather than $H_i.\lookup(\tau,\kappa^j_i)$. Also, all reads and writes, as well as the PIR queries, are made now to $\calS_0$ and $\calS_1$.

\paragraph{Reshuffles.} The key modification made to the scheme lays in the reshuffling procedure. The reshuffles are still performed in the same frequency. However, the roles of the servers change, as only two servers participate in the protocol. First, $\calS_0$ prepares all headers of blocks that have to be reshuffled into the destination hash table, and, together with the client, invokes the appropriate  oblivious $\build$ procedure to hash the blocks into a new hash table.

Now that the headers are shuffled, it remains to match them to the data. The matching procedure begins with tagging the headers. $\calS_0$ sends the shuffled headers, one by one to the client. The client decrypts every header, computes its tag using a new epoch, and sends the tag back to $\calS_0$. The headers corresponding to empty slots in the hash table are tagged using numbered values, e.g. '$empty\circ 1$'. Notice that the number of empty slots in the hash table and stash, combined, is fixed and independent of the input. Next, $\calS_1$ sends the records (headers and data) that correspond to the shuffled headers, one by one, in a random order. Among the actual records, $\calS_1$ also sends as many (numbered) empty records as required to match the number of empty records in the newly-reshuffled hash table. The client tags every record he receives from $\calS_1$, and forwards it $\calS_0$ together with its tag. $\calS_0$ can easily match every record he receives to a header in the hash table or stash, according to the tags. Once this is complete, $\calS_0$ sends the new hash table and stash to $\calS_1$, through the client. Refer to Algorithm~\ref{algo:2reshuffle} for full details.

\begin{algorithm}[h!t!b]
	\caption{Two-Server Construction: Reshuffle}
	\label{algo:2reshuffle}
	\begin{algorithmic}[1]
		\item[\textbf{Reshuffling headers into table $T^j_{i+1}$}]
		\parState{$\calS_0$ sends all records in levels $1,\dots, i$ and the stash, one by one, to the client. The client re-encrypts every record he receives and forwards it to $\calS_1$, while eliminating all empty records. $\calS_1$ inserts every record he receives to a temporary array $Y$. Server $\calS_1$ now sends every header in $Y$ back to $\calS_0$, through the client.}
		\parState{Let $\hat{Y}$ be the array of encrypted headers received by $\calS_0$. The client generates a fresh hashing key $\kappa^j_i\gets H_i.\gen(N)$, and, together with $\calS_0$, invokes $(\hat{T},\hat{S})\gets H_i.\build(\kappa^j_i,\hat{Y})$ to obliviously hash the headers into a hash table and stash.}
		\item[]
		\item[\textbf{Matching data to headers.}]
		\parState{$\calS_0$ sends $(\hat{T},\hat{S})$, record by record, to the client. The client decrypts every header $v$ he receives, and computes a tag $\tau\gets F_s(i + 1,j, e^j_{i+1}, v)$. If the header is empty, then $\tau\gets F_s(i + 1,j, e^j_{i+1}, empty\circ z)$, where $z$ is a counter that starts at 1 and and increments after every empty header. Notice that the number of empty headers, denoted by $Z$, depends only on $i$. The client sends the tag back to $\calS_0$.} \label{step:tagnewhash}
		\parState{$\calS_1$ inserts $Z$ empty records $(empty\circ 1,\cdot),\dots,(empty\circ Z,\cdot)$ to $Y$. Server $\calS_1$ permutes $Y$ randomly, and sends it, one record at a time, to the client.}
		\parState{The client re-encrypts every record $(v,x)$ it receives, and sends it to $\calS_0$ with a tag $\tau$, that is the output of $F_s$ on $v$ with the appropriate epoch.}\label{step:tagrecords}
		\parState{$\calS_0$ matches every tagged record it receives to one of the tags it received in Step~\ref{step:tagnewhash}, and inserts the corresponding record to its appropriate slot (either in $\hat{T}$ or $\hat{S}$).}\label{step:matching}
		\parState{At this point, $\calS_0$ holds the newly reshuffled hash table and stash, headers and data. The tags are discarded. $\calS_0$ sends both the table and the stash to $\calS_1$, via the client. Both servers replace the old stash and $T^j_{i+1}$ with the new data.}
	\end{algorithmic}
\end{algorithm}
\subsection{From Two Servers to $m$ Servers}

We generalize the ideas behind the two-server construction from Section~\ref{sec:2servers} to construct a family of $(m,m-1)$-ORAM schemes. The construction of the two-server scheme consists of two main parts: queries and reshuffles. In order to transform the construction to the multi-server setting, we transform each of these components while maintaining their security.

\paragraph{Query using multi-server PIR.} To obliviously simulate a query to a block, the client follows the protocol used in the two-server construction (Algorithm~\ref{algo:3query}). However, now that we want to achieve privacy against any colluding subset of corrupt servers, we use an $m$-server PIR protocol which guarantees such a privacy. That is, instead of invoking two-server PIRs to query blocks from the stash and hierarchy levels, the client now uses an $(m,m-1)$-PIR protocol involving all $m$ servers, where the joint view of any $m-1$ servers is (computationally) independent of the target index. In particular, we can use the straight-forward $m$-server generalization of the basic PIR protocol from\cite{PIR}. Since this protocol, as well as many known $m$-server PIRs, follow the standard PIR setting where the data is assumed to be replicated in all of the servers, the servers during the ORAM execution will hold identical replicates of the same data structure.

\paragraph{Extending the matching procedure.} Reshuffles of levels are done in the same frequency, and in a very similar manner as in the two-server protocol. The only change we make is in matching procedure. To match the content to the tags, we cannot rely only on two servers, since they might be both corrupt. Instead, we let all servers participate. The reshuffling procedure from Algorithm~\ref{algo:2reshuffle} is followed up to Step~\ref{step:tagrecords}. After the client receives the permuted records from $\calS_1$, he re-encrypts them and forwards them to $\calS_2$. $\calS_2$, in its turn, randomly permutes the records it receives, and forwards them to $\calS_3$ (if it exists), through the client. This continues until all servers, except $\calS_0$, have received the records and permuted them on their own. Once they all had, the client tags the records and sends them to $\calS_0$, who matches them to the shuffled headers, as described in Steps~\ref{step:tagrecords} and~\ref{step:matching}. Lastly, the new final hash table and stash are sent to all servers.

We note that different servers can be used to reshuffle different levels, thus distributing the load of work and communication equally among all servers.

\subsection{Analysis} \label{sec:multianalysis}
\subsubsection{Complexity.}
The query complexity of the $m$-server scheme is identical to that of the three-server construction, and is equal to $O(k+L)$. To obliviously construct a hash table and a stash for a level $i$, the client and the servers exchange $C_\build(H_i)=O(\beta d^{i-1}k)$ records (recall $\beta:=\max_i \frac{C_\build(H_i)}{d^{i-1}k}$). However, since the build-up is done over tags of size $\log N$ bits, rather than whole blocks of size $\Omega(\beta\log N)$, the cost of $O(\beta d^{i-1}k)$ tags translates to $O(d^{i-1}k)$ overhead in blocks. The matching procedure also has a linear cost in the size of the level, that is $O(m_i)$. Since the reshuffling of the level occurs every $d^{i-1}k$ rounds, this amortizes to $O(1+m_i/d^{i-1}k)$ overhead per level, and $O(L+\sum_{i=1}^L \frac{m_i}{d^{i-1}k})$ overall.
\subsubsection{Security.}
Following Definition~\ref{def:m-oram}, it suffices to prove the following Lemma.
\begin{lemma}[Security of the $m$-server ORAM]\label{lem:msec}
	Let $\view_a(\vec{y})$ be the view of server $\calS_a$ during the execution of the $m$-server ORAM protocol, described in Section~\ref{sec:2multiservers}, over a virtual access pattern $\vec{y}=((v_1,x_1),\dots,(v_\ell,x_\ell))$. For any subset of servers $A\subseteq \{0,\dots,m-1\}$ of size $|A|\leq m-1$, there exist a simulator $\simulator_A$, such that for every $\vec{y}$ of length $\ell$, the distributions $\simulator_A(\ell)$ and $\left<\view_a(\vec{y})\right>_{a\in A}$ are computationally indistinguishable.
\end{lemma}

We make the same assumptions taken in the analysis of the three-server scheme: encryption is secure, and the tagging functions are random.

First, consider the view of the servers at the reshuffles in the two-server case. Claims~\ref{claim:stash} and~\ref{claim:reshuff} are true for the multi-server scheme as well, therefore, the amount of encrypted data exchanged in Step 1 of Algorithm~\ref{algo:2reshuffle} is oblivious, and therefore so is the view of the servers. From Definition~\ref{def:oblihash}, the view seen in Step 2 can be also simulated by producing an access pattern for an arbitrary execution of the oblivious $H_i.\build$ procedure. As for the matching procedure (Steps 3-7), the view of $\calS_1$ consists of the newly constructed hash table and stash, both encrypted and of fixed size, and therefore can be simulated arbitrarily. $\calS_0$ receives a sequence of tags computed using $F_s$. We show below that these tags are computed for unique headers.

\begin{claim} \label{claim:2uniquetagsreshuff} The tagging function $F_s(\cdot)$ will not be computed twice on the same input in Step~3 of Algorithm~\ref{algo:2reshuffle} throughout the executions of the algorithms during the ORAM simulation.
\end{claim}
\begin{proof}
	For headers $v$ that are of the form $'empty\circ z'$ the claim is trivial. It is also easy to verify that Claim~\ref{claim:3uniquetags} is true for the $m$-server ORAM construction as well, thus implying the claim for non-empty headers.\qed
\end{proof}

Since unique, the sequence of tags seen by the server is indistinguishable from uniform distinct values, and $\simulator_0$ simulates them as such. Lastly, $\calS_0$ receives a sequence of tagged records. The records themselves are encrypted, and therefore can be simulated. The tags were obtained by tagging the same set of headers already tagged previously, however, in an order that is generated by $\calS_1$. Since $\calS_1$ applies a uniformly random permutation, that is not known to $\calS_0$ and is \emph{independent of previous reshuffles}, over the headers, the view of $\calS_0$ in Step 5 of the algorithm can be simulated by taking the sequence of tags generated previously by $\simulator_0$, and permute them randomly.

In the multi-server setting, we take the seemingly hardest case in which $\calS_0$ is among the $m-1$ corrupt servers. To simulate the view of $\calS_0$, up to Step~\ref{step:tagrecords} of the algorithm, we follow the steps taken above. The combined view of all other corrupt servers consists of sequences of encrypted records of fixed lengths. We argue that the sequence of tagged records, received by $\calS_0$ once all servers took their turn, can be simulated by randomly permuting the tags that were seen previously in Step~\ref{step:tagnewhash}. The key observation to make here is that, since there is at least one honest server that permutes the records randomly, using a permutation unknown to the adversary, the order of the tagged records is random from the adversary's point of view, and is independent of the input. More specifically, every permutation that the honest server may choose yields a different order of the records. Since the permutation is chosen uniformly at random, the order of the records distributes uniformly.

The access pattern seen by the servers during the queries is even easier than the three-server construction from Section~\ref{sec:3servers}, since now we rely on the obliviousness of the $\lookup$ procedure (see Definition~\ref{def:oblihash}), rather than on the randomness of the tags (as in Claim~\ref{claim:3tagsareuniform}). From Definition~\ref{def:oblihash}, the sequence of $H_i.\lookup(v,\kappa^j_i)$ values is indistinguishable from a sequence generated for an arbitrary sequence of addresses $v$ using random hash keys. Thus, $\simulator_A$  simulates the queries to the hash tables by generating random keys using $\gen$, and computing $\lookup$ over an arbitrary sequence of addresses. The transcripts of the $(m-1)$-private PIR protocols, can be simulated from the security definition of PIR.

\section{Conclusion and Open Questions}
In this paper, we presented a family of efficient distributed ORAM schemes, that achieve both sub-logarithmic overhead, and commit to a small block size. Although our parameters match the lower bound proven in\cite{PIRORAM} for PIR-based ORAM, we use techniques that are not captured in the model considered in their proof. It remains an open question whether we can further develop these techniques to surpass their bound, without using additively homomorphic encryption (AHE), or linear server work.

By allowing the servers to perform a linear amount of computation, we have shown that constant overhead ORAM is achievable. Our 4-server construction is simple, and does not use a complex data structure for the server storage, but rather stores it as a plain array. The simplicity of the construction opens the possibility to use it as a building block, to further reduce the overhead of distributed ORAM protocols, in exchange for increasing server-side computations.

%\bibliographystyle{alpha}
%\bibliography{dist_oram_bib}{}

\begin{thebibliography}{SvDS{\etalchar{+}}13}
	
	\bibitem[ACMR95]{TWOTIER}
	M.~Adler, S.~Chakrabarti, M.~Mitzenmacher, and L.~Rasmussen.
	\newblock Parallel randomized load balancing.
	\newblock In {\em Proceedings of the Twenty-seventh Annual ACM Symposium on
		Theory of Computing}, STOC '95, pages 238--247, New York, NY, USA, 1995. ACM.
	
	\bibitem[AFN{\etalchar{+}}17]{PIRORAM}
	I.~Abraham, {C. W.} Fletcher, K.~Nayak, B.~Pinkas, and L.~Ren.
	\newblock {\em Asymptotically tight bounds for composing ORAM with PIR}, volume
	10174 LNCS of {\em Lecture Notes in Computer Science (including subseries
		Lecture Notes in Artificial Intelligence and Lecture Notes in
		Bioinformatics)}, pages 91--120.
	\newblock Springer Verlag, Germany, 2017.
	
	\bibitem[AIK06]{CRYPTONC0}
	B.~Applebaum, Y.~Ishai, and E.~Kushilevitz.
	\newblock Cryptography in $NC^0$.
	\newblock {\em SIAM J. Comput.}, 36(4):0097-5397, December 2006.	
	
	\bibitem[AKL+18]{OPTORAMA}
	G.~Asharov, I.~Komargodski, W-K.~Lin, K.~Nayak, and E.~Shi.
	\newblock Optorama: Optimal oblivious RAM.
	\newblock Cryptology ePrint Archive, Report 2018/892, 2018.
	
	\bibitem[AKS83]{AKS}
	M.~Ajtai, J.~Koml\'{o}s, and E.~Szemer{\'e}di.
	\newblock An 0(n log n) sorting network.
	\newblock In {\em Proceedings of the Fifteenth Annual ACM Symposium on Theory
		of Computing}, STOC '83, pages 1--9, New York, NY, USA, 1983. ACM.
	
	\bibitem[AKST14]{AKST}
	D.~Apon, J.~Katz, E.~Shi, and A.~Thiruvengadam.
	\newblock Verifiable oblivious storage.
	\newblock In Hugo Krawczyk, editor, {\em Public-Key Cryptography -- PKC 2014},
	pages 131--148, Berlin, Heidelberg, 2014. Springer Berlin Heidelberg.
	
	\bibitem[BGI15]{FSS}
	E.~Boyle, N.~Gilboa, and Y.~Ishai.
	\newblock Function secret sharing.
	\newblock In Elisabeth Oswald and Marc Fischlin, editors, {\em Advances in
		Cryptology - EUROCRYPT 2015}, pages 337--367, Berlin, Heidelberg, 2015.
	Springer Berlin Heidelberg.
	
	\bibitem[BIW07]{PIRBIW}
	O.~Barkol, Y.~Ishai, and E.~Weinreb.
	\newblock On locally decodable codes, self-correctable codes, and t-private
	{PIR}.
	\newblock In M.~Charikar, K.~Jansen, O.~Reingold, and J.~D.~P. Rolim, editors,
	{\em Approximation, Randomization, and Combinatorial Optimization. Algorithms
		and Techniques}, pages 311--325, Berlin, Heidelberg, 2007. Springer Berlin
	Heidelberg.
	
	\bibitem[CGKS98]{PIR}
	B.~Chor, O.~Goldreich, E.~Kushilevitz, and M.~Sudan.
	\newblock Private information retrieval.
	\newblock {\em J. ACM}, 45(6):965--981, November 1998.
	
	\bibitem[CGLS17]{CGLS}
	T-H.~H. Chan, Y.~Guo, W-K. Lin, and E.~Shi.
	\newblock Oblivious hashing revisited, and applications to asymptotically
	efficient {ORAM} and {OPRAM}.
	\newblock Cryptology ePrint Archive, Report 2017/924, 2017.
	
	\bibitem[CKN$^+$18]{PERFECT3ORAM}
	T-H.~H.~Chan, J.~Katz, K.~Nayak, A.~Polychroniadou, and E.~Shi.
	\newblock More is less: Perfectly secure oblivious algorithms in the multi-server setting. 
	\newblock Cryptology ePrint Archive, Report 2018/851, 2018.
	
	\bibitem[CMS99]{PIRCMS}
	C.~Cachin, S.~Micali, and M.~Stadler.
	\newblock Computationally private information retrieval with polylogarithmic
	communication.
	\newblock In Jacques Stern, editor, {\em Advances in Cryptology --- EUROCRYPT
		'99}, pages 402--414, Berlin, Heidelberg, 1999. Springer Berlin Heidelberg.
	
	\bibitem[DG16]{PIRDG}
	Z.~Dvir and S.~Gopi.
	\newblock 2-server {PIR} with subpolynomial communication.
	\newblock {\em J. ACM}, 63(4):39:1--39:15, September 2016.
	
	\bibitem[DS17]{DS}
	J.~Doerner and A.~Shelat.
	\newblock Scaling {ORAM} for secure computation.
	\newblock In {\em Proceedings of the 2017 ACM SIGSAC Conference on Computer and
		Communications Security}, CCS '17, pages 523--535, New York, NY, USA, 2017. ACM.
	
	\bibitem[DvF{\etalchar{+}}16]{ONION}
	S.~Devadas, M.~{van Dijk}, {C. W.} Fletcher, L.~Ren, E.~Shi, and D.~Wichs.
	\newblock {\em Onion ORAM: A constant bandwidth blowup oblivious RAM}, volume
	9563 of {\em Lecture Notes in Computer Science (including subseries Lecture
		Notes in Artificial Intelligence and Lecture Notes in Bioinformatics)}, pages
	145--174.
	\newblock Springer Verlag, Germany, 1 2016.
	
	\bibitem[FNR{\etalchar{+}}15]{BUCKET}
	C.~W. Fletcher, M.~Naveed, L.~Ren, E.~Shi, and E.~Stefanov.
	\newblock Bucket {ORAM}: Single online roundtrip, constant bandwidth oblivious
	ram.
	\newblock {\em IACR Cryptology ePrint Archive}, 2015:1065, 2015.
	
	\bibitem[GI14]{DPF}
	N.~Gilboa and Y.~Ishai.
	\newblock Distributed point functions and their applications.
	\newblock In Phong~Q. Nguyen and Elisabeth Oswald, editors, {\em Advances in
		Cryptology -- EUROCRYPT 2014}, pages 640--658, Berlin, Heidelberg, 2014.
	Springer Berlin Heidelberg.
	
	\bibitem[GKK{\etalchar{+}}12]{MPCORAM}
	S.~D. Gordon, J.~Katz, V.~Kolesnikov, F.~Krell, T.~Malkin, M.~Raykova, and
	Y.~Vahlis.
	\newblock Secure two-party computation in sublinear (amortized) time.
	\newblock In {\em Proceedings of the 2012 ACM Conference on Computer and
		Communications Security}, CCS '12, pages 513--524, New York, NY, USA, 2012.
	ACM.
	
	\bibitem[GM11]{GM}
	M.~T. Goodrich and M.~Mitzenmacher.
	\newblock Privacy-preserving access of outsourced data via oblivious {RAM}
	simulation.
	\newblock In Luca Aceto, Monika Henzinger, and Ji{\v{r}}{\'i} Sgall, editors,
	{\em Automata, Languages and Programming}, pages 576--587, Berlin,
	Heidelberg, 2011. Springer Berlin Heidelberg.
	
	\bibitem[GO96]{GO}
	O.~Goldreich and R.~Ostrovsky.
	\newblock Software protection and simulation on oblivious RAMs.
	\newblock {\em J. ACM}, 43(3):431--473, May 1996.
	
	\bibitem[Gol87]{G}
	O.~Goldreich.
	\newblock Towards a theory of software protection and simulation by oblivious
	RAMs.
	\newblock In {\em Proceedings of the Nineteenth Annual ACM Symposium on Theory
		of Computing}, STOC '87, pages 182--194, New York, NY, USA, 1987. ACM.
	
	\bibitem[Goo11]{SHELL}
	M.~T. Goodrich.
	\newblock Randomized shellsort: A simple data-oblivious sorting algorithm.
	\newblock {\em J. ACM}, 58(6):27:1--27:26, December 2011.
	
	\bibitem[GR05]{PIRGR}
	C.~Gentry and Z.~Ramzan.
	\newblock Single-database private information retrieval with constant
	communication rate.
	\newblock In Lu{\'i}s Caires, Giuseppe~F. Italiano, Lu{\'i}s Monteiro, Catuscia
	Palamidessi, and Moti Yung, editors, {\em Automata, Languages and
		Programming}, pages 803--815, Berlin, Heidelberg, 2005. Springer Berlin
	Heidelberg.
	
	\bibitem[IKOS08]{IKOS}
	Y.~Ishai, E.~Kushilevitz, R.~Ostrovsky, and A.~Sahai.
	\newblock Cryptography with constant computational overhead.
	\newblock In {\em Proceedings of the Fortieth Annual ACM Symposium on Theory of
		Computing}, STOC '08, pages 433--442, New York, NY, USA, 2008. ACM.
	
	\bibitem[KLO12]{KLO}
	E.~Kushilevitz, S.~Lu, and R.~Ostrovsky.
	\newblock On the (in)security of hash-based oblivious RAM and a new balancing
	scheme.
	\newblock In {\em Proceedings of the Twenty-third Annual ACM-SIAM Symposium on
		Discrete Algorithms}, SODA '12, pages 143--156, Philadelphia, PA, USA, 2012.
	Society for Industrial and Applied Mathematics.
	
	\bibitem[KMW09]{CUCKOOSTASH}
	A.~Kirsch, M.~Mitzenmacher, and U.~Wieder.
	\newblock More robust hashing: Cuckoo hashing with a stash.
	\newblock {\em SIAM J. Comput.}, 39(4):1543--1561, December 2009.
	
	\bibitem[KO97]{PIRKO}
	E.~Kushilevitz and R.~Ostrovsky.
	\newblock Replication is not needed: Single database, computationally-private
	information retrieval.
	\newblock In {\em Proceedings of the 38th Annual Symposium on Foundations of
		Computer Science}, FOCS '97, pages 364--, Washington, DC, USA, 1997. IEEE
	Computer Society.
	
	\bibitem[LN18]{LN}
	K.~G.~Larsen and J.~B.~Nielsen.
	\newblock Yes, there is an oblivious RAM lower bound! 
	\newblock In Hovav Shacham and Alexandra Boldyreva, editors, {\em Advances in Cryptology – CRYPTO 2018}, pages 523–-542, Cham, 2018. Springer International Publishing.
	
	\bibitem[LO13]{LO}
	S.~Lu and R.~Ostrovsky.
	\newblock Distributed oblivious RAM for secure two-party computation.
	\newblock In Amit Sahai, editor, {\em Theory of Cryptography}, pages 377--396,
	Berlin, Heidelberg, 2013. Springer Berlin Heidelberg.
	
	\bibitem[MBM15]{CHF}
	T.~Moataz, E.~Blass, and T.~Mayberry.
	\newblock Chf-oram: A constant communication ORAM without homomorphic
	encryption.
	\newblock Cryptology ePrint Archive, Report 2015/1116, 2015.
	
	\bibitem[OS97]{OS}
	R.~Ostrovsky and V.~Shoup.
	\newblock Private information storage (extended abstract).
	\newblock In {\em Proceedings of the Twenty-ninth Annual ACM Symposium on
		Theory of Computing}, STOC '97, pages 294--303, New York, NY, USA, 1997. ACM.
	
	\bibitem[Ost90]{O}
	R.~Ostrovsky.
	\newblock Efficient computation on oblivious RAMs.
	\newblock In {\em Proceedings of the Twenty-second Annual ACM Symposium on
		Theory of Computing}, STOC '90, pages 514--523, New York, NY, USA, 1990. ACM.
	
	\bibitem[PPRY18]{PANORAMA}
	S.~Patel, G.~Persiano, M.~Raykova, and K.~Yeo. 
	\newblock Panorama: Oblivious RAM with logarithmic overhead. \newblock Cryptology ePrint Archive, Report 2018/373, 2018.
	
	\bibitem[PR04]{CUCKOO}
	R.~Pagh and F.~F. Rodler.
	\newblock Cuckoo hashing.
	\newblock {\em J. Algorithms}, 51(2):122--144, May 2004.
	
	\bibitem[PR10]{ORAMCUCKOOPR}
	B.~Pinkas and T.~Reinman.
	\newblock Oblivious RAM revisited.
	\newblock In {\em Proceedings of the 30th Annual Conference on Advances in
		Cryptology}, CRYPTO'10, pages 502--519, Berlin, Heidelberg, 2010.
	Springer-Verlag.
	
	\bibitem[RYF{\etalchar{+}}13]{PROCESSORS}
	L.~Ren, X.~Yu, C.~W. Fletcher, M.~van Dijk, and S.~Devadas.
	\newblock Design space exploration and optimization of path oblivious RAM in
	secure processors.
	\newblock {\em SIGARCH Comput. Archit. News}, 41(3):571--582, June 2013.
	
	\bibitem[SCSL11]{TREE}
	E.~Shi, T.~H.~H. Chan, E.~Stefanov, and M.~Li.
	\newblock Oblivious RAM with O((logn)3) worst-case cost.
	\newblock In Dong~Hoon Lee and Xiaoyun Wang, editors, {\em Advances in
		Cryptology -- ASIACRYPT 2011}, pages 197--214, Berlin, Heidelberg, 2011.
	Springer Berlin Heidelberg.
	
	\bibitem[SvDS{\etalchar{+}}13]{PORAM}
	E.~Stefanov, M.~van Dijk, E.~Shi, C.~W. Fletcher, L.~Ren, X.~Yu, and
	S.~Devadas.
	\newblock Path oram: An extremely simple oblivious RAM protocol.
	\newblock In {\em Proceedings of the 2013 ACM SIGSAC Conference on Computer
		Communications Security}, CCS '13, pages 299--310, New York, NY, USA, 2013.
	ACM.
	
	\bibitem[WCS15]{CIRCUIT}
	X.~Wang, H.~Chan, and E.~Shi.
	\newblock Circuit ORAM: On tightness of the Goldreich-Ostrovsky lower bound.
	\newblock In {\em Proceedings of the 22Nd ACM SIGSAC Conference on Computer and
		Communications Security}, CCS '15, pages 850--861, New York, NY, USA, 2015.
	ACM.
	
	\bibitem[WGK18]{WGK}
	X.~Wang, D.~Gordon, and J.~Katz.
	\newblock Simple and efficient two-server ORAM.
	\newblock Cryptology ePrint Archive, Report 2018/005, 2018.
	
	\bibitem[ZMZQ16]{ZMZQ}
	J.~Zhang, Q.~Ma, W.~Zhang, and D.~Qiao.
	\newblock Mskt-oram: A constant bandwidth ORAM without homomorphic encryption.
	\newblock Cryptology ePrint Archive, Report 2016/882, 2016.
	
\end{thebibliography}

\newcommand{\etalchar}[1]{$^{#1}$}

\appendix
\section{From Amortized Overhead to Worst-Case Overhead}\label{app:deamortize}
Up to this point, in all of our constructions and analysis, we referred to the amortized overhead an ORAM protocol incurs. Although we have achieved an overall sub-logarithmic overhead, some operations, specifically reshuffles of large levels, might require sending $\Omega(N)$ blocks in a round. This might be unwanted in several applications, in which we may have to guarantee that every query, without exceptions, incurs a bounded worst-case amount of overhead. Ostrovsky and Shoup\cite{OS} were the first to construct a worst-case ORAM scheme. They achieved a $O(\log^3 N)$ worst-case overhead by de-amortizing the classic hierarchical solution from\cite{GO}. Since then, many works have tackled the worst-case version of the ORAM problem, either by de-amortizing existing schemes (e.g. \cite{KLO}), or by constructing new tree-based ORAMs\cite{TREE,PORAM}.

In this section, we show how to de-amortize our two-server construction from Section~\ref{sec:2servers} in order to obtain a two-server oblivious RAM scheme with an equal \emph{worst-case} overhead. The same method can be used also to de-amortize the $(m,m-1)$-ORAM schemes, however, it fails when applied to the three-server scheme from Section~\ref{sec:3servers}.
\subsection{Overview}
We follow the techniques used in\cite{OS} and \cite{KLO}. The idea is to 'spread' the reshuffles and perform them throughout many rounds, rather than in a single round. Once a level has to be reshuffled, it is not reshuffled immediately, instead, a \emph{reshuffling process} starts:
\begin{enumerate}
	\item The full level is put aside and marked as 'inactive', while the protocol continues as usual with a new empty 'active' instance of the level.
	\item In the round when the process starts, and in every round to follow, a fixed amount of $\rho$ blocks is exchanged between the client and servers as part of the reshuffle. The blocks are shuffled into an empty hash table in a larger active level. Up until the completion of the process, the destination hash table is not accessed and is treated as if it was empty.
	
	During the process, the client might have to query the inactive levels for blocks that reside there. Therefore, during the search for a block, PIR protocols are performed over the inactive levels as well.
	\item The process ends when the reshuffle is complete and all blocks reside in a newly created hash table in both servers.
\end{enumerate}

The rate of the process, i.e. $\rho$, has to be chosen wisely. On the one hand, $\rho$ has to be small enough so the overall amortized overhead, caused by the reshuffling of all levels, does not exceed $O(k+L)$. On the other hand, very small $\rho$ means slow long-lasting reshuffling processes, which might lead to storage overflow and undesired overhead, as reshuffles will overlap and stack over time. In fact, for every level $i=1,\dots,L$, we choose a different rate $\rho_i$, that depends on the parameters of the hashing scheme $H_i$, and that guarantees that a reshuffling process of level $i$ ends before a new one starts. Thus, it would be sufficient to have only two instances of every level, alternating between active and inactive.

This approach introduces few obstacles, for which we suggest the following solutions.
\paragraph{Reshuffles are Unproportional to Level Size.} Take the top level as an example. The top level has to be reshuffled every $k$ queries. The level is not always reshuffled into the next level, but it may be reshuffled, together with other blocks, to much larger ones (levels of size up to $\Theta(N)$). Reshuffles to large levels clearly require large amount of communication, therefore require large number of rounds to complete, during which, the top level is reshuffled many times. This situation forces us to maintain lots of inactive instances of the top level, as it is reshuffled very frequently, while the reshuffles last for unproportionally long periods.

To solve this problem, we reshuffle levels separately of each other: level $i$ is always reshuffled into the next empty hash table in level $i+1$ (when all its hash tables are full). Furthermore, we add an extra hash table to every level in the hierarchy, making it $d$ tables in a level. Conceptually, these two modifications do not change much in the scheme: levels are still reshuffled in the same frequency, and the reshuffling of levels $1,\dots, i$ to level $i+1$ is replaced by a cascading series of reshuffles, starting from the top level, and ending with level $i$ reshuffled into level $i+1$. However, this change is necessary for the de-amortization.

\paragraph{Inactive Levels are Immutable.} In the amortized two-server scheme, whenever a block is found and read from the hierarchy, it is wiped out and marked as a dummy block, before it is re-written to the top level. To maintain the integrity of the reshuffling process, we clearly cannot alter the content of an inactive level while it is being reshuffled. Therefore, we allow having duplicated blocks in the hierarchy, and we avoid wiping out blocks after they are queried.
\paragraph{Duplicates Compromise Security.} The uniqueness of tags in the hierarchy played an essential role in the security of the scheme in Section~\ref{sec:2servers}. Now that we allow duplicates, the scheme is no longer secure, as repeated tags, which indicate identical block addresses, might appear. We suggest the following solution. We add a single bit to the header of every block that indicates whether its content is up to date. The blocks inserted to the top level are marked as \emph{up-to-date}. Every time a block is read, it is marked as \emph{outdated}. Before the start of the reshuffling process of any level, the client eliminates outdated blocks, and thus we guarantee uniqueness of tags in every hash table, which is, in fact, sufficient for the security of the scheme.
\subsection{De-amortization in Details}
We modify the scheme from Section~\ref{sec:2servers} as follows.
\paragraph{Data Structure.} The  following changes are made to the data structure:
\begin{itemize}
	\item[--] An additional hash table is added to every level.
	\item[--] We double every level in the hierarchy into two instances: one initially marked as 'active', the other as 'inactive'. The two instances are identical in structure, but they operate separately, and each has its own hashing key.
\end{itemize}

\paragraph{Block Headers.} In addition to the virtual address, we require that the header of every block contains an \emph{up-to-date} bit, that indicates whether the block is up-to-date and its content is relevant. The bit is encrypted together with the virtual address, however, it is not involved in the computation of the tags. The notion of a block is extended to ($v,u,x$), with $u\in\{\text{\emph{up-to-date}},\text{\emph{outdated}}\}$ being the up-to-date bit.
\paragraph{Queries.} The query procedure (Algorithm~\ref{algo:3query}) is worst-case efficient. Nonetheless, in order to support the new reshuffling method, we make a couple of modifications:
\begin{itemize}
	\item[--] The target block is searched for in the inactive levels as well. Clearly, out-to-date blocks are ignored during the scan.
	\item[--] Instead of wiping out blocks that are found in the hierarchy, and replacing them with dummy hash blocks (Algorithm~\ref{algo:3query}, Step~\ref{step:reuploadheaders}), the client marks them as unread, by toggling the \emph{up-to-date} bit off.
\end{itemize}
\paragraph{Reshuffles.} As a first step towards de-amortized reshuffles, we make the following changes.
\begin{itemize}
	\item[--] A level is reshuffled once its all hash tables are full.
	\item[--] A level is always reshuffled to the next empty hash table in the subsequent level (the first bullet guarantees that there is at least one empty hash table).
\end{itemize}
Notice that, as a consequence to these changes, the last hash table of every level is never accessed while it is active, as the level is immediately reshuffled once it is filled. Thus, the purpose of adding an extra hash table is to 'make room' for blocks that are pending to be reshuffled to a larger level.

The reshuffling of a level begins with marking it as inactive, and replacing it by the other instance of the level, which is guaranteed to be clear and already reshuffled to the next level, as we show later on (Lemma~\ref{lem:reshufoverlap} below).

The reshuffling proceeds following the reshuffling procedure of the scheme in Section~\ref{sec:2servers}, and is performed over many rounds. In every round of the reshuffling procedure of level $i$, $\rho_{i}$ blocks are exchanged between the client and servers. We choose $\rho_{i}=\Theta(\frac{m_i}{d^{i-1}k})$ (exact value is provided in the analysis section).

Lastly, the reshuffling procedure is adjusted to handle outdated blocks. Just before the reshuffling begins, the client replaces all outdated blocks, that reside in the level, with numbered dummy hash blocks. More specifically, $\calS_0$ (or $\calS_1$) sends all records in the level, one by one, to the client. The client decrypts every record $(v,u,x)$ it receives. If $u=\text{\emph{up-to-date}}$, then the client re-encrypts the record and sends it back to both $\calS_0$ and $\calS_1$. Otherwise, the client replaces the record with $(dummyHash\circ r,\cdot,\cdot)$, where $r$ is a counter that increments whenever a dummy hash record is created (this is the only place it happens, as we do not create dummy blocks in the queries), and sends it back to the servers. The servers, in their turn, replace the received records with the records in the level, and proceed with the reshuffling.

Although the reshuffling processes of a level do not overlap, reshuffling processes of different levels might be running simultaneously. The chosen values of $\rho_i$ guarantee that this does not cause an undesired blow-up in the overhead of the protocol. However, now that the client has to maintain the state of many reshuffling processes in parallel, his working memory is no longer constant. To overcome this issue, the client encrypts the state of every ongoing reshuffling process, and stores it in one of the servers. When a reshuffling process is resumed, the client downloads its state, and proceeds with the reshuffling. Once $\rho_i$ blocks were exchanged, the client pauses the process, and uploads the new state to the server. Thus, at every moment, the client holds the state of a single reshuffling process, which is of constant size.

\subsection{Analysis}
\begin{claim}\label{claim:reshufoverlap}
	For a carefully chosen value $\rho_i\in\Theta(\frac{m_i}{d^{i-1}k})$, the reshuffling processes of level $i-1$ in the de-amortized two-server scheme do not overlap.
\end{claim}
\begin{proof}
	During the reshuffling into a hash table in level $i$, $O(d^{i-1}k + m_i)$ blocks are exchanged between the client and servers (see analysis in Section~\ref{sec:multianalysis}), specifically, there exists a constant $c_i$, s.t. the cost of the reshuffling is bounded by $c_i(d^{i-1}k + m_i)$. Therefore, the reshuffling process into level $i$ lasts less than $\frac{c_i}{\rho_i}(d^{i-1}k + m_i)$ rounds, and is initiated every $d^{i-1}k$ rounds. Hence, for $\rho_i > \frac{c_im_i}{d^{i-1}k}+1$ the processes do not overlap.\qed
\end{proof}

\paragraph{Query Overhead.} It was shown in Section~\ref{sec:2servers} that the worst-case overhead incurred by the queries in the amortized two-server scheme is $O(k+L)$, the changes made to the query procedure do not change that.

\paragraph{Reshuffling Overhead.} The reshuffling of  level $i-1$ to a hash table in level $i$ contributes an amount of $\rho_i$ blocks to the worst-case overhead. Summing over all levels, we get $O(\sum_{i=1}^L \frac{m_i}{d^{i-1}k})$.

\section{Application to Secure Multi-Party Computation}\label{app:mpc}
In this section, we briefly describe how our distributed ORAM constructions can be used in the setting of secure multi-party computation of RAM programs to achieve efficient protocols. The idea of using oblivious RAM for secure computation was suggested in prior work\cite{OS,LO,MPCORAM}. Ostrovsky and Shoup\cite{OS} were the first to propose a method to use oblivious RAM protocols with simple client circuits (e.g. \cite{KLO}), to construct secure computation protocols for functions that are naturally presented as RAM programs, rather than circuits, and thus avoiding the heavy cost caused by unrolling such programs into circuits. The main drawback of their approach is the fact that they use single-server ORAM as a building block, over which they apply two-server PIR\cite{PIR}. That is, they do not take full advantage of the existence of two servers. Lu and Ostrovsky\cite{LO} make a better use of the two servers and achieve a better overhead.

The main idea behind using distributed ORAM for secure computation is as follows. To securely compute a RAM program in an $m$-party setting, we use an $m$-server ORAM protocol. Each of the servers is simulated by one of the parties. From the security of the underlying ORAM, the secrecy of the inputs is protected from the views of the servers. However, it is not guaranteed that the client's view is oblivious in the inputs, and therefore, the client has to be simulated securely. Specifically, the client's state (its local memory) is secret-shared among the $m$-parties. To simulate the next query in the ORAM protocol, the parties invoke a secure computation protocol to simulate the client's circuit, and compute the query to the servers, and the client's next state, given the current state. In all of our constructions, the local computations performed by the client can be described by circuits of size linear in the required word size, therefore, using a secure computation protocol that incurs a constant overhead in the size of the circuit (e.g. \cite{IKOS}) does not cause an undesired blowup.

By using any of our constructions, together with a secure circuit computation protocol with the appropriate number of parties and privacy, we achieve the following results. We refer the reader to \cite{OS,LO} for further details.

\begin{theorem}
	Suppose there exists a symmetric-key encryption scheme and a hash function modeled as a random function or an efficient PRF. Suppose that for any constant $m\geq 2$, there exists an $m$-party secure circuit computation protocol with constant overhead, that is private against an adversary corrupting an arbitrary subset of parties (e.g. \cite{IKOS}). Then, to securely compute a RAM program that runs in $T(n)$ time with access to $S(n)$ space, with program size (including inputs) bounded by $\Lambda(n)$, for input of size $n$ words, there exist
	\begin{itemize}
		\item[--] a three-party secure RAM computation protocol in the semi-honest model with $O(\omega(1)\cdot\log (T+\Lambda)/\log\log (T+\Lambda))$ multiplicative overhead in communication and computational complexity. To achieve optimal overhead, the protocol requires that the program runs over data of $\Omega(\log^{1+\epsilon} n)$-bit words (or, alternatively, $O(\log (T+\Lambda)/\log\log (T+\Lambda))$ overhead with $\Omega(\log^{1.5+\epsilon}n)$-bit words).
		\item[--] for every constant $m\geq2$, an $m$-party secure RAM computation protocol in the semi-honest model, that is private against an adversary corrupting an arbitrary subset of parties, with $O(\log (T+\Lambda)/\log\log (T+\Lambda))$ multiplicative overhead in communication and computation complexity. To achieve optimal overhead, the protocol requires that the program runs over data of $\Omega(\log^{2} n)$-bit words.
		\item[--] a four-party secure RAM computation protocol in the semi-honest model with constant multiplicative overhead in communication and $O(S(n))$ overhead in computation complexity. To achieve optimal overhead, the protocol requires that the program runs over data of $\Omega(\lambda\log n)$-bit words, where $\lambda$ is a security parameter. In this construction we also assume the existence of constant-depth PRGs (with minimal expansion)~\cite{CRYPTONC0}.
	\end{itemize}
	
	all protocols incur an additive one-time cost of $O(\Lambda\log \Lambda/\log\log \Lambda)$ for setup.
\end{theorem}

\section{Hashing}\label{app:hashing}
We hereby formally define a hashing scheme, and an oblivious hashing scheme. As will be shown below, these definitions capture many of the known hashing schemes, especially those applicable to ORAM.
\begin{definition}\label{def:hash}
	A \emph{$(n,m,s)$-Hashing Scheme} is a tuple $H=(\gen,\build,\lookup)$, where $\gen$ and $\lookup$ are PPT algorithms, and $\build$ is a program in the standard RAM model (see Section~\ref{sec:pre}). The scheme operates in two phases:
	\begin{enumerate}
		\item\textbf{Build-up.}
		\begin{align*}
		&k\gets H.\gen(N)\\
		&(T,S)\gets H.\build(k,X)
		\end{align*}
		The input to $H.\build$ is a key, $k$, generated by $H.\gen$, and a set of $n$ distinctly tagged values $X=\{(\tau,x)\mid \tau\in [N]\}$. The output consists of:
		\begin{itemize}
			\item An array $T$ of size $m$, referred to as the \emph{hash table}. Every position in $T$ is either empty or contains some record $(\tau,x)$. 
			\item A \emph{stash} $S$, that contains at most $s$ tagged values,
		\end{itemize}
		The build-up might fail with probability at most $N^{-\omega(1)}$ over the choice of $k$. The \textbf{build-up complexity} of $H$, denoted by $C_{\build}(H)$ is the amount of data blocks exchanged between the CPU and RAM in an $H.\build$ invocation.
		\item\textbf{Queries.}
		\begin{align*}
		&Q\gets H.\lookup(\tau,k)
		\end{align*}
		The input to $H.\lookup$ is a tag $\tau\in[N]$, and a key $k$. The output of $H.\lookup$ is a set of positions $Q\subseteq [m]$. Let $(T,S,k)$ be the output of a previously performed construction phase on input $(k,X)$. If $(\tau,x)\in X$ for some $x$, then either $(\tau,x)$ resides in $T[j]$ for some $j\in Q$, or $(\tau,v)$ resides in the stash $S$. The \textbf{query complexity} of $H$, denoted by $C_{\lookup}(H)$, is the maximal size of a set $Q$ returned by $H.\lookup$.
	\end{enumerate}
\end{definition}

\begin{definition}\label{def:oblihash}
	An \emph{Oblivious Hashing Scheme} $H$ is a hashing scheme where
	\begin{itemize}
		\item $H.\build$ is oblivious in $k,X$. That is, for any two distinct inputs, $(k,X)$ and $(k',X')$ of the same length, the corresponding access patterns produced by $H.\build$ are computationally indistinguishable.
		\item For any two sequences of distinct tags $\vec{\tau}=(\tau_1,\dots,\tau_\ell)$ and $\vec{\tau'}=(\tau'_1,\dots,\tau'_\ell)$, the outputs of $H.\lookup$ on $\vec{\tau}$ and $\vec{\tau'}$ are computationally indistinguishable.
	\end{itemize}
\end{definition}

For completeness, we provide below an overview of the hash schemes that were considered in the oblivious RAM context.

\paragraph{Standard hashing.} To store $n$ tagged values, standard hashing generates a hash table that is a plain array of $m=cn$ buckets (e.g. $c=2,4$), each of size $O(\log n)$, and a hash function $h_k:[N]\to\{1,\dots,m\}$. A record $(\tau,x)$ is stored in bucket $h_k(\tau)$ in the array. When $h_k$ is thought of as a random oracle, the probability of an overflow is negligible in $N$. In their classic hierarchical solution, Goldreich and Ostrovsky\cite{GO} constructed an oblivious variant of standard hashing with build-up complexity of $O(n\log n)$, by using PRFs as hash functions, and oblivious sort\cite{AKS,SHELL} to obliviously store the data in the hash table. Lu and Ostrovsky\cite{LO} showed that by using a (shared) stash of size $O(\log N)$, the bucket size can be reduced to $O(\log N/\log\log N)$ while maintaining a negligible failure probability.

\paragraph{Cuckoo hashing.} Cuckoo hashing\cite{CUCKOO} is a variant of standard hashing, that was found to be useful in some ORAM schemes\cite{ORAMCUCKOOPR,GM,KLO,LO}. In cuckoo hashing, every bucket may contain at most one record, and two hash functions $h_k^0,h_k^1$ are used. A record $(\tau,x)$ is initially stored in $h_k^0(\tau)$, and is moved to $h_k^1(\tau)$ once another record is inserted into $h_k^0(\tau)$. Notice that this strategy might cause a long chain, or even a cycle of block relocations. In such case, new functions are chosen, and the structure is updated. It was shown that the amortized insertion time in cuckoo hashing is $O(1)$. A variant of cuckoo hashing is stashed cuckoo hashing, presented by Kirsch et al.\cite{CUCKOOSTASH}. In stashed cuckoo hashing, a stash of size $s$ is used to handle overflows. They show that using a constant size stash, inserting $n$ elements into a table of size $\Theta(n)$ will fail with probability $O(n^{-s})$. Goodrich and Mitzenmacher\cite{GM} showed that a table of size $\Omega(\log^7 N)$ with a stash of size $O(\log N)$ fails with probability negligible in $N$. Oblivious sorting was shown to be useful also to construct oblivious cuckoo hashing scheme with build-up complexity of $O(n\log n)$\cite{ORAMCUCKOOPR,GM,CGLS}.

\paragraph{Two-tier hashing.} In two-tier hashing\cite{TWOTIER}, as the name suggests, the data is stored in two standard hash tables, each with a different hash function. Every table contains $O(n/b)$ buckets, each of capacity $b=\log^\epsilon N$ for $\epsilon\in(0.5,1)$. A record $(\tau,v)$ is inserted to the appropriate bucket in the first tier, if the bucket is full, then the record is inserted to the second tier. Chan et al.\cite{CGLS} presented an oblivious variant of two-tier hashing with build-up complexity of $O(n\log n)$, and used it to construct oblivious RAM. Although their ORAM scheme has an overhead asymptotically equal to the ORAM based on cuckoo hashing\cite{KLO}, its advantage lays in the simplicity of oblivious two-tier hashing, both in construction and analysis, relative to oblivious cuckoo hashing.

\section{Single-server Query Procedure\cite{KLO}}\label{app:klo}
Below is a pseudo-code of the query procedure of the single-server ORAM scheme from~\cite{KLO}, which is described in a high level in Section~\ref{sec:klo}.
\begin{algorithm}[H]
	\caption{Single-Server Balanced Hierarchy: Query}
	\label{algo:1query}
	\begin{algorithmic}[1]
		\State Allocate a local register of the size of a single entry.
		\State Initialize a flag $\found\gets 0$.
		\parState{Download the top level and stash $S$, record by record. If a record $(v,x)$ was found, then store it in the local register, and mark $\found\gets 1$.}
		\For{every level $i=1\dots L$}
		\State $t' \gets \floor{t/k}$
		\State $r\gets \floor{(t' \bmod d^i) / d^{i-1}}$
		\For{every hash table $j=r\dots 1$}
		\parState{If not $\found$, $Q\gets H_i.\lookup(v,\kappa^j_i)$. Else, $Q\gets H_i.\lookup(dummy \circ t,\kappa^j_i)$.}
		\parState{Download records $T^j_i[Q]$, record by record. If a record $(v,x)$ was found, then store it in the local register, and mark $\found\gets 1$.}
		\EndFor
		\EndFor
		\State If the query is a write query, overwrite $x$ in the register.
		\State Read each entry of the entire top level from both servers one at a time, re-encrypt it, then write it back, with the following exception: if the the entry $(v,x)$ was first found at the top level, then overwrite $x$ with the (possibly) new value from the register, otherwise, write $(v,x)$ in the first empty spot of the form $(empty,\cdot)$.
		\State Increment the counter $t$, and reshuffle the appropriate levels. 
	\end{algorithmic}
\end{algorithm}

\section{Proof of Claim~\ref{claim:reshuff}}\label{app:claims}
In this section, we prove Claim~\ref{claim:reshuff} from Section~\ref{sec:3analysis}. We re-state it here for convenience.

\reshufflingclaim*

\begin{proof}
	We prove the lemma using induction. For $t=0$, we have $i^*=0$ and $r^t_i=0$ for every $i$, the entire storage is empty, therefore the lemma is trivially correct. Suppose the lemma holds true for some $t=t'\cdot k\geq0$. Let $i^*$ be the smallest value of $i$ for which $r^{t+k}_i > 0$. Notice that $i^*$ is also the largest integer $i$ for which $d^{i-1}\mid t'+1$ (the level to which the reshuffle of round $t+k$ is done). We proceed by proving the lemma is correct at round $t+k$ for four different groups of levels:
	\begin{enumerate}
		\item \emph{The top level}: from the induction hypothesis, the top level is empty when the reshuffle at round $t$ is finished. Performing queries $t,\dots,t+k$ fills up the top level entirely (as it is not reshuffled down during this time). The reshuffle after query $t+k$ empties the top level, thus,~\ref{item:reshufftop} holds.
		\item \emph{Levels $1\leq i< i^*$}: the reshuffle at round $t+k$ evicts all blocks from levels smaller than $i^*$ to hash table $T^{r_{i^*}}_{i^*}$. Therefore, once the reshuffle is done, levels $i=1,\dots,i^*-1$ are empty, matching the fact that $r^{t+k}_i=0$ and proving the lemma for such $i$'s.
		\item \emph{Level $i^*$}: from the choice of $i^*$, there exists some $0<r<d$ for which $t' + 1=q\cdot d^{i^*}+r\cdot d^{i^*-1}$. It also holds that $r_{i^*}^{t+k}=r$, and $r_{i^*}^{t}=\floor{(t'\bmod d^{i^*})/d^{{i^*}-1}}=\floor{(r\cdot d^{i^*-1}-1)/d^{i^*-1}}=r-1$. These give $r_{i^*}^{t+k}=r_{i^*}^{t}+1$. By induction, we can imply that after the reshuffle of round $t$, $T^{r}_{i^*}$ is empty. Furthermore, all hash tables in levels $i<i^*$ are all full prior to the reshuffle, i.e. contain $d^{i-1}(k+s)$ blocks, since $r^t_i = d - 1$ for such $i$'s. The reshuffle procedure in round $t+k$ takes all blocks in levels $i=0,\dots,i^*-1$ and stash, and inserts them into $T^{r}_{i^*}$ and the stash. The count of these blocks is $s+k + \sum_{i=1}^{i^*-1} d^{i-1}(k+s)=d^{i^*-1}(k+s)$, making $T^{r}_{i^*}$ full (notice that blocks that are sent to the stash are replaced by dummy stash blocks). Since the other tables in level $i^*$ are untouched in  rounds $t$ and $t+k$, the lemma holds for $i^*$.
		\item \emph{Levels $i^*<i\leq L$}: these levels are not involved in the reshuffle at round $t+k$, and remain unchanged during queries $t,\dots,t+k$. Thus, it suffices to show that $r^t_i=r^{t+k}_i$ for every such $i$. Again, we rewrite $t' + 1=q\cdot d^{i^*}+r\cdot d^{i^*-1}$ for some $0<r<d$. Take some $i>i^*$. Using the fact that $r>0$, we get $r^t_i=\floor{(t'\bmod d^{i})/d^{{i}-1}}=\floor{((q\cdot d^{i^*}\bmod d^{i})+r\cdot d^{i^*-1} - 1)/d^{{i}-1}}=\floor{((q\cdot d^{i^*}\bmod d^{i})+r\cdot d^{i^*-1})/d^{{i}-1}}=r^{t+k}_i$.
		
		\end{enumerate}	\qed
	\end{proof}

\end{document}